\documentclass[12pt]{amsart} 
\usepackage{amsaddr}
\usepackage[margin=1in]{geometry}
\usepackage{paralist}
\usepackage{amsmath,amssymb}
\usepackage[usenames,dvipsnames]{xcolor}
\usepackage{tikz}
\usepackage{graphicx}
\usepackage{booktabs}
\usepackage{multirow}
\usepackage{array}
\usepackage{setspace}
\usepackage{amsthm}
\usepackage{amsfonts}
\usepackage{mathtools}
\usepackage{subfig}
\usepackage[authoryear,longnamesfirst,sort]{natbib}
\usepackage{psfrag}
\usepackage{verbatim}
\usepackage[mathlines]{lineno} 

\bibpunct{(}{)}{,}{}{}{,}

\newcommand{\p}{\prime}

\DeclareMathOperator*{\E}{E}

\theoremstyle{plain}

\newtheorem{proposition}{Proposition}
\newtheorem{corollary}{Corollary}
\newtheorem{lemma}{Lemma}

\theoremstyle{definition}

\newtheorem*{unmarkedassumption}{Assumption}

\newtheorem{assumptionSPM}{Assumption}

\newtheorem{assumptionSBM}{Assumption}

\newtheorem{assumptionSPMIV}{Assumption}

\newtheorem{assumptionSPMTS}{Assumption}

\newtheorem{assumptionQSPMIV}{Assumption}

\newtheorem{assumptionSI}{Assumption}

\makeatletter
\newcommand{\subjclassname@JEL}{JEL classification}
\makeatother

\theoremstyle{remark}

\numberwithin{equation}{section}

\allowdisplaybreaks[4]

\begin{document} 

\title{Complementarity and Identification} 
\author{Tate Twinam}
\address{Interdisciplinary Arts and Sciences \\ University of Washington Bothell}
\date{\today} 
\email{twinam@uw.edu}
\keywords{Partial identification, nonparametric bounds, supermodularity, instrumental variables, program evaluation}
\subjclass[JEL]{C21, C26}
\thanks{I am indebted to Arie Beresteanu for his support and feedback on multiple drafts. Additionally, constructive comments from two anonymous referees and the editors Victor Chernozhukov and Peter Phillips substantially improved the quality of the paper.}

\maketitle

\newpage

\begin{abstract}

This paper examines the identification power of assumptions that formalize the notion of complementarity in the context of a nonparametric bounds analysis of treatment response. I extend the literature on partial identification via shape restrictions by exploiting cross--dimensional restrictions on treatment response when treatments are multidimensional; the assumption of \textit{supermodularity} can strengthen bounds on average treatment effects in studies of policy complementarity. This restriction can be combined with a statistical independence assumption to derive improved bounds on treatment effect distributions, aiding in the evaluation of complex randomized controlled trials. Complementarities arising from treatment effect heterogeneity can be incorporated through \textit{supermodular instrumental variables} to strengthen identification in studies with one or multiple treatments. An application examining the long--run impact of zoning on the evolution of urban spatial structure illustrates the value of the proposed identification methods.

\end{abstract}

\newpage

\doublespacing
\allowdisplaybreaks[4]

\begin{section}{Introduction}\label{intro} 

Complementarities arise naturally in many economic problems, often manifesting as policy interactions or treatment effect heterogeneity among observed subgroups of a population. This paper examines how assumptions that formalize the notion of complementarity can aid in the identification of treatment effects. The analysis employs a nonparametric bounds approach, where identification is driven by qualitative restrictions rooted in economic theory or empirical evidence rather than strong functional form or unconfoundedness assumptions. This approach will yield interval estimates of parameters of interest; however, informative bounds are often preferable to precise (but wrong) estimates obtained under incorrect assumptions. Partial identification tools have been fruitfully applied to a wide range of empirical problems.\footnote{Examples include: \cite{giustinelli2011non} and \cite{okumura2014concave} on the returns to education, \cite{kreider2007disability} on disability and employment, \cite{bhattacharya2012treatment, bhattacharya2008treatment} on the mortality effects of Swan--Ganz catheterization, \cite{kreider2009partially} on the effect of universal health insurance on medical expenditures,  \cite{pepper2000intergenerational} on the intergenerational transmission of welfare receipt, \cite{manski1998bounding} on sentencing and recidivism, and \cite{gundersen2012impact} on the health effects of the National School Lunch Program.}

In particular, I explore the identification power yielded by assuming that individual treatment response functions exhibit \textit{supermodularity} when treatments are multidimensional. This assumption allows one to construct more informative bounds in studies of policy complementarity, which are typically stymied by the absence of pseudo--experimental variation in the assignment of multiple treatments. Complementarities arising from interactions between treatment effects and observable covariates can be formalized as \textit{supermodular instrumental variables} to improve bounds on average treatment effects. This novel instrumental variable approach is broadly applicable to studies with one or multiple treatments. Complementarity is frequently invoked in economics, but studies of its identification power have been limited to very specific contexts. This paper develops general results applicable to program evaluation in a wide range of empirical situations, and their value is illustrated by an empirical application on the long--run effects of zoning on urban spatial structure.

Typically, empirical studies seek to estimate the effect of a single treatment on one or more outcome variables. However, the effect of a treatment may vary substantially with the value of other (endogenously--determined) treatment variables. When policymakers have multiple tools at their disposal, understanding how different policies enhance or offset each other is crucial. If the positive impact of some policy intervention is substantially larger when combined with a second (costly) intervention, a measure of the magnitude of this difference is necessary for a proper cost--benefit analysis. The supermodularity and submodularity assumptions proposed here can aid in quantifying how policy impacts differ with the associated policy environment.

For example, unemployment relief is a multidimensional policy, involving a choice of both potential benefit duration and the wage replacement rate. \cite{lalive2006how} show both theoretically and empirically that these two dimensions are complementary, with simultaneous increases in both the replacement rate and potential benefit duration leading to an increase in unemployment duration substantially larger than the sum of the effects measured individually for particular subgroups. The \citeauthor{lalive2006how} study exploits variation in both dimensions of unemployment relief that has the characteristics of a natural experiment, but such opportunities are very rare. Pseudo--experimental variation along multiple policy dimensions is far less common than similar variation in individual policies. This has arguably led to the overwhelming focus on the effects of policies in isolation. The partial identification tools developed here, which are applicable in the absence of any unusual pseudo--experimental policy variation, should enhance the ability of researchers to measure treatment effect heterogeneity due to policy complementarities in a wide range of contexts.\footnote{The sensitivity of effects to the surrounding policy environment may partly explain the wide variation in estimates of treatment effects for similar policies in different contexts found in many literatures; see, for example, the discussion in \cite{lalive2006how} on the effects of unemployment benefit policies on re--employment rates.} 

Relatedly, responses to a treatment may differ among subpopulations defined by observable covariates. Many recent experimental studies have discussed the importance of treatment effect heterogeneity between subgroups \citep{djebbari2008heterogeneous, bitler2006mean, bitler2008distributional, bitler2014can, feller2009beyond}. Qualitative information about such treatment effect heterogeneity leads naturally to supermodular instrumental variables, which can help narrow the bounds on average treatment effects in the same manner as a traditional instrumental variable or a monotone instrumental variable.\footnote{See \cite{manski2000monotone} and section \ref{instrumental variables}.} Supermodular instrumental variables can be applied in the case of a single treatment or multiple treatments, making them a potentially valuable addition to the range of identifying assumptions available to applied researchers.

While the bulk of the paper focuses on identification using non--experimental data, the assumptions developed in this paper can be applied in the evaluation of complex randomized controlled trials (RCTs) involving multiple treatments.   The structural assumptions introduced can be used to obtain stronger bounds on the (partially identified) distribution of treatment effects. Since average treatment effects are identified in this context, the supermodularity or submodularity of average effects can be established, and this can be used to provide some justification for the stronger structural assumptions. Similarly, the validity of supermodular instrumental variable assumptions can be established and used to justify stronger \textit{quantile supermodular instrumental variable} assumptions, which can also be applied in the case of a single treatment. 

The literature on partial identification is extensive.\footnote{See \cite{manski2003partial} for a comprehensive overview.} Many of the contributions of Charles Manski and coauthors are relevant to the results developed below and are reviewed as appropriate. The literature on complementarity and identification is relatively small. \cite{molinari2008identification} connect supermodularity to identification in the context of game estimation. They show that the approach of \cite{aradillas2008identification} applies to games with supermodular payoff functions. \cite{eeckhout2011identifying} find that they cannot identify (using wage data alone) whether or not the technology of a firm is supermodular, i.e., whether or not more productive workers sort towards more productive jobs. \cite{graham2014complementarity} analyze how reallocations of indivisible heterogeneous inputs across production units (leaving a potentially complementary input fixed) may affect average output. They discuss identification and estimation of the effects of a variety of correlated matching rules. \cite{lazzati2014treatment} uses monotone comparative statics to partially identify treatment response in the presence of endogenous social interactions. The shape restrictions proposed have previously been used in the context of estimation to improve efficiency; \cite{beresteanu2005nonparametric, beresteanu2007nonparametric} considers the efficiency gains from imposing a variety of restrictions, including supermodularity and submodularity.

\end{section}

\begin{section}{Notation and Setup}\label{Notation and Setup}

Individuals are drawn from a population $I$. The set $I$, the Borel $\sigma$--algebra of subsets of $I$ denoted by $\mathcal{I}$, and the probability measure $P$ together form a probability space $\left(I, \mathcal{I}, P\right)$. Every individual $i\in I$ is associated with a vector of covariates $x_{i} \in X$ and a vector of realized treatments $z_{i} \in T$, where $T$ is the treatment set.\footnote{I reserve superscripts to refer to vector components.} Since I focus on the identification of treatment effects in the presence of multiple treatments, the following assumption is imposed on the structure of the treatment space:

\begin{unmarkedassumption}
The treatment space $T$ is such that
\begin{itemize}

\item $T \subseteq \mathbb{R}^{L}$ with $L\in\mathbb{N}$,

\item $T$ is partially ordered under the product order,\footnote{If $t_{1}$ and $t_{2}$ are incomparable, I write $t_{1} \parallel t_{2}$.} and

\item $T$ is a nonempty lattice. 

\end{itemize}
\end{unmarkedassumption}

The lattice assumption means that, for any $t_{1}, t_{2}\in T$, $T$ contains the join (least upper bound) of $t_{1}$ and $t_{2}$, denoted by $t_{1} \vee t_{2}$, as well as the meet (greatest lower bound) of $t_{1}$ and $t_{2}$, denoted by $t_{1} \wedge t_{2}$.\footnote{The meet and join operations depend on the particular order imposed on the lattice; for example, the join of $(2,0)$ and $(1,1)$ in $\mathbb{R}^2$ is equal to $(2,1)$ under the product order and $(2,0)$ under the lexicographic order.} Examples of lattices include $\mathbb{R}^2$, $\mathbb{Z}\times\mathbb{R}$, and $\left\{0,1\right\}^n$ for $n\in\mathbb{N}$. An element $t_{1}$ of a lattice $T$ is the top (bottom) of $T$ if $t_{2} \leq t_{1}$ $\left(t_{1} \leq t_{2}\right)$ for all $t_{2}\in T$; if $t_{1}$ is not the top or bottom, it is in the interior. If the top (or bottom) of a lattice exists, it is unique. A subset $S\subseteq T$ is a sublattice of $T$ if, for any $t_{1}, t_{2}\in S$, $S$ contains the meet and join of $t_{1}$ and $t_{2}$ in $T$. The advantage of the lattice assumption is the notational clarity it provides when employing supermodularity and submodularity assumptions. In this paper, I restrict attention to discrete treatments, as these are most commonly encountered in practice. Dimensions of the treatment may be binary or multi--valued \citep{cattaneo2010multi}. 

Every individual $i$ is associated with a (measurable) response function $y_{i}\left(\cdot\right): T \rightarrow Y\in\mathbb{R}$ mapping treatments into outcomes $y_{i}\left(t\right)$.\footnote{I suppress $i$ when referring to arbitrary response functions, covariates, or realized treatments.} $z_{i} \in T$ is the treatment that $i$ actually receives, so $y_{i}\left(z_{i}\right)$ is individual $i$'s realized outcome, $\left\{y_{i}\left(t\right)\right\}_{t\neq z_{i}}$ are individual $i$'s counterfactual outcomes, and $\left\{y_{i}\left(t\right)\right\}_{t\in T}$ are individual $i$'s potential outcomes.\footnote{The stable unit treatment value assumption (alternatively referred to as noninterference by \cite{cox1958planning} and individualistic treatment response by \cite{manski2013identification}) is maintained throughout the paper; this states that individuals' potential outcomes $\left\{y\left(t\right)\right\}_{t\in T}$ do not depend on other individuals' realized treatments \citep{rubin1978bayesian}.} Throughout, I assume that there exist  $\underline{K},\overline{K}\in\mathbb{R}$ such that $\underline{K} \leq y\left(t\right) \leq \overline{K}$ for all $t$; these are global bounds on response functions. In the absence of these global bounds, the results below will generally be uninformative. All well--defined expectations are assumed to exist.\footnote{If an expectation $\E\left[\, y\left(t_{1}\right) \mid z=t_{2} \,\right]$ is ill--defined because the event $z=t_{2}$ is off the support of $z$, I establish the convention that $\E\left[\, y\left(t_{1}\right) \mid z=t_{2} \,\right]P\left(z=t_{2}\right) \equiv 0$.}

\end{section}

\begin{section}{Shape Restrictions}\label{shape restrictions}

In this section, I explore the identifying power of shape restrictions that formalize complementarity and substitutability, with an emphasis on the identification of average treatment effects. Shape restrictions proposed in the previous literature are reviewed before moving on to the novel restrictions proposed here. Using these assumptions, I derive bounds on average treatment effects for both simple and complex treatment spaces. 

\cite{manski1989anatomy} introduced the no--assumption bounds on $E\left[\, y\left(t\right) \,\right]$. The no--assumption upper bound is the average of $E\left[\, y\left(t\right) \mid z = t \,\right]$ and the global upper bound $\overline{K}$, weighted respectively by $P\left(z=t\right)$ and $P\left(z\neq t\right)$; likewise for the lower bound. Since they are typically wide, research has focused on other credible assumptions that yield additional identifying power. 

\cite{manski1997monotone} studied the identification power of assumptions on the shape of individual response functions; in particular, he considered restricting response functions to be monotone, semi--monotone, or concave--monotone. Semi--monotone treatment response (SMTR), which I employ below, requires response functions to be weakly increasing in $t$. SMTR has the same identification power regardless of whether $T\subseteq \mathbb{R}$ or $T\subseteq \mathbb{R}^L$ for $L>1$, except that in the latter case, it is possible that $t_{1} \parallel t_{2}$. \cite{bhattacharya2008treatment} derives bounds using SMTR without assuming a particular direction of monotonicity. \cite{okumura2014concave} study the identification power of concave--monotone treatment response combined with monotone treatment selection (discussed in section \ref{instrumental variables}). 

SMTR is a within--dimension restriction on the response functions. Additional identification power can be obtained from cross--dimension restrictions, where the marginal effect of a change in some dimensions of the treatment variable depends on the values of the other dimensions:

\begin{assumptionSPM}[Supermodularity]\label{SPM}
Response functions are \textit{supermodular} on a sublattice $S\subseteq T$ if, for all $t_{1}, t_{2}\in S$,
\begin{equation}\label{spm}
y\left(t_{2}\right) + y\left(t_{1}\right) \leq y\left(t_{1} \vee t_{2}\right) + y\left(t_{1} \wedge t_{2}\right) 
\end{equation}
\end{assumptionSPM}

\begin{assumptionSBM}[Submodularity]\label{SBM}
Response functions are \textit{submodular} on a sublattice $S\subseteq T$ if, for all $t_{1}, t_{2}\in S$,
\begin{equation}\label{sbm}
y\left(t_{2}\right) + y\left(t_{1}\right) \geq y\left(t_{1} \vee t_{2}\right) + y\left(t_{1} \wedge t_{2}\right)
\end{equation}
\end{assumptionSBM}

\ref{SPM} is a formalization of the notion of complementarity. If two dimensions of a treatment $t = \left(t^{1},t^{2}\right)$ are complementary, then the magnitude of the change in the response variable due to an increase in the first dimension $t^{1}$ is increasing with $t^{2}$. Thus, each component of the treatment amplifies the marginal effect of the other component. In the case of a linear model
\begin{equation}\label{linear}
y = \alpha + \beta t^{1} + \delta t^{2} + \gamma t^{1} t^{2}
\end{equation}
supermodularity is equivalent to the sign restriction $\gamma \geq 0$. More generally, if $y$ is a (sufficiently smooth) nonlinear function of the treatment, supermodularity can be interpreted as a nonnegativity restriction on the cross--partial derivative $\frac{\partial^2 y}{\partial t^{1} \partial t^{2}}$. \ref{SBM} is a formalization of substitutability, the case where elements of the treatment may mitigate each others effects. In the linear model above, submodularity is equivalent to the sign restriction $\gamma\leq 0$. If both supermodularity and submodularity hold, response functions are said to be \textit{modular}. Since assumptions \ref{SPM} and \ref{SBM} can be applied on sublattices of $T$, it is possible to allow some dimensions of a treatment to be complements while those same dimensions are substitutes with other dimensions. 

\cite{neumark2011does} provide an example of policy complementarity in a study on the interaction between the Earned Income Tax Credit (EITC) and the minimum wage. They find that a higher minimum wage enhances the positive effect of the EITC on the labor supply of single mothers; they find the opposite effect for childless individuals, suggesting a crowding--out effect. These findings suggest that assumptions \ref{SPM} and \ref{SBM}, respectively for each subgroup, could be applied in other studies on how the effect of minimum wage changes are influenced by the EITC or similar programs. Another naturally multidimensional policy is zoning. Zoning laws typically regulate many aspects of the built environment; most broadly, they regulate both what types of uses are allowed (commercial, industrial, etc.) and how densely land can be developed (lot coverage of buildings, maximum height, etc.). The effects of specific zoning policies may vary with the overall policy bundle, which I explore further in the empirical application in section \ref{emp}.

The amount one can learn about $E\left[\, y\left(t_{1}\right) \,\right]$ or $E\left[\, y\left(t_{1}\right) - y\left(t_{2}\right) \,\right]$ from the data alone depends on $P\left(z=t_{1}\right)$ and, in the latter case, $P\left(z=t_{2}\right)$. If $P\left(z=t_{1}\right)$ is small, the data are practically uninformative about $E\left[\, y\left(t_{1}\right) \,\right]$.\footnote{In the case where one or more of the dimensions of the treatment are continuous, the data are necessarily uninformative about almost all of the treatments. This motivates my restriction to discretely--valued treatments.} Thus, the researcher faces a trade--off where richer treatment spaces (which entail a larger number of treatments) allow for more interesting questions but generally lead to less precise answers. Adding ``nuisance'' dimensions to the treatment space that allow for the application of additional \ref{SPM} or \ref{SBM} assumptions will generally not aid in the identification of treatment effects of interest. 

In propositions \ref{spm bounds 0} and \ref{spm bounds 2}, I show how \ref{SPM} and \ref{SBM} can be used to compute bounds on the expectations of average treatment effects. In general, these bounds will improve upon the no--assumption bounds in the case of multidimensional treatments; with only a single treatment, \ref{SPM} and \ref{SBM} have no identifying power. The simplest nontrivial lattice treatment space is $T = \left\{\left(0,0\right), \left(1,0\right), \left(0,1\right), \left(1,1\right)\right\}$, which corresponds to two binary treatments. The following result shows the implications of supermodularity for identification on this simple treatment space:
\begin{proposition}\label{spm bounds 0}
Assume that $T = \left\{\left(0,0\right), \left(1,0\right), \left(0,1\right), \left(1,1\right)\right\}$. Assume that \ref{SPM} holds on $T$. Then, the bounds
\begin{gather}
\E\left[\, y\left(1,0\right) \mid z=\left(1,0\right) \,\right] P\left(z = \left(1,0\right)\right) + \underline{K}P\left(z \neq \left(1,0\right)\right) \nonumber \\
- \E\left[\, y\left(0,0\right) \mid z=\left(0,0\right) \,\right] P\left(z = \left(0,0\right)\right) - \overline{K} P\left(z \neq \left(0,0\right)\right) \nonumber \\
\leq \E\left[\, y\left(1,0\right) - y\left(0,0\right) \,\right] \leq \label{spm zero 1} \\
\E\left[\, y\left(1,1\right) \mid z=\left(1,1\right) \,\right] P\left(z = \left(1,1\right)\right) + \E\left[\, y\left(1,0\right) \mid z=\left(1,0\right) \,\right] P\left(z = \left(1,0\right)\right) \nonumber \\ 
+ \overline{K} P\left(z\in \left\{\left(0,0\right), \left(0,1\right)\right\}\right) - \E\left[\, y\left(0,1\right) \mid z=\left(0,1\right) \,\right] P\left(z = \left(0,1\right)\right) \nonumber \\
- \E\left[\, y\left(0,0\right) \mid z=\left(0,0\right) \,\right] P\left(z = \left(0,0\right)\right)  - \underline{K}P\left(z\in\left\{\left(1,0\right), \left(1,1\right)\right\}\right) \nonumber
\end{gather}
and 
\begin{gather}
\E\left[\, y\left(1,1\right) \mid z=\left(1,1\right) \,\right] P\left(z = \left(1,1\right)\right) + \E\left[\, y\left(1,0\right) \mid z=\left(1,0\right) \,\right] P\left(z = \left(1,0\right)\right) \nonumber \\ 
+ \underline{K} P\left(z\in\left\{\left(0,0\right), \left(0,1\right)\right\}\right) - \E\left[\, y\left(0,1\right) \mid z= \left(0,1\right) \,\right] P\left(z = \left(0,1\right)\right) \nonumber \\ 
- \E\left[\, y\left(0,0\right) \mid z=\left(0,0\right) \,\right] P\left(z = \left(0,0\right)\right) - \overline{K} P\left(z\in\left\{\left(1,0\right), \left(1,1\right)\right\}\right) \nonumber \\
\leq \E\left[\, y\left(1,1\right) - y\left(0,1\right) \,\right] \leq \label{spm zero 2} \\
\E\left[\, y\left(1,1\right) \mid z=\left(1,1\right) \,\right] P\left(z = \left(1,1\right)\right) + \overline{K}P\left(z \neq \left(1,1\right)\right)  \nonumber \\
- \E\left[\, y\left(0,1\right) \mid z=\left(0,1\right) \,\right] P\left(z = \left(0,1\right)\right) - \underline{K} P\left(z \neq \left(0,1\right)\right) \nonumber 
\end{gather}
are sharp.\footnote{There is no guarantee that these bounds will be nonempty; if an assumption implies that the bounds on the parameter of interest are empty, the assumption can be falsified by the data. This caveat applies to all the results that follow.} The no--assumption bounds remain sharp for $\E\left[\, y\left(1,1\right) - y\left(0,0\right) \,\right]$, $\E\left[\, y\left(1,0\right) - y\left(0,1\right) \,\right]$, $\E\left[\, y\left(0,1\right) - y\left(1,0\right) \,\right]$, and each average potential outcome $\E\left[\, y\left(\cdot\right) \,\right]$ defined on $T$.
\end{proposition}

\begin{proof}[Proof of proposition \ref{spm bounds 0}]

First, I show that \ref{SPM} does not improve upon the no--assumption bounds on potential outcomes. \ref{SPM} implies that 
\[
y\left(1,0\right) + y\left(0,1\right) \leq y\left(1,1\right) + y\left(0,0\right)
\]
For each $i$, exactly one of these outcomes is observed. The unobserved terms may take any value in $\left[\underline{K}, \overline{K}\right]$. Consider bounding the average potential outcome $\E\left[\, y\left(1,0\right) \,\right]$. When $z \neq \left(1,0\right)$, there are three cases to consider. If $z = \left(1,1\right)$, then \ref{SPM} implies
\[
y\left(1,0\right) \leq \overline{K} \leq y\left(1,1\right) + \overline{K} - \underline{K} 
\]
If $z = \left(0,1\right)$, then
\[
y\left(1,0\right) \leq \overline{K} \leq \overline{K} + \overline{K} - y\left(0,1\right)
\]
If $z = \left(0,0\right)$, then
\[
y\left(1,0\right) \leq \overline{K} \leq \overline{K} + y\left(0,0\right) - \underline{K}
\]
All three of these inequalities are implied without \ref{SPM}, so the assumption is nonbinding. Thus, it follows that
\begin{equation*}
y\left(1,0\right) \in
\begin{cases}
\left\{y\left(1,0\right)\right\} & \text{if } z = \left(1,0\right) \\
\left[\underline{K}, \overline{K}\right] & \text{if } z \in \left\{\left(0,0\right), \left(0,1\right), \left(1,1\right)\right\}
\end{cases}
\end{equation*}
Taking expectations yields the no--assumption bounds. A similar argument applies to the other elements of $T$.

The \ref{SPM} inequality does permit strengthened identification results for treatment effects. In the no--assumption case, if $z = \left(1,1\right)$ or $z = \left(0,1\right)$, then $y\left(1,0\right) - y\left(0,0\right) \in \left[\underline{K} - \overline{K}, \overline{K} - \underline{K}\right]$. Under \ref{SPM}, the fact that we observe one of $\left\{y\left(1,1\right),y\left(0,1\right)\right\}$ allows us to further reduce this upper bound. Sharp bounds for the treatment effects $y\left(1,0\right) - y\left(0,0\right)$, $y\left(1,1\right) - y\left(0,1\right)$, and $y\left(1,1\right) - y\left(0,0\right)$ are given below.

\begin{equation}\label{spm zero 1 proof}
y\left(1,0\right) - y\left(0,0\right) \in 
\begin{cases}
\left[\underline{K} - y\left(z\right), \overline{K} - y\left(z\right)\right] & \text{if } z = \left(0,0\right) \\
\left[y\left(z\right) - \overline{K}, y\left(z\right) - \underline{K}\right] & \text{if } z = \left(1,0\right) \\
\left[\underline{K} - \overline{K}, \overline{K} - y\left(z\right)\right] & \text{if } z = \left(0,1\right) \\
\left[\underline{K} - \overline{K}, y\left(z\right) - \underline{K}\right] & \text{if } z = \left(1,1\right) 
\end{cases}
\end{equation}

\begin{equation}\label{spm zero 2 proof}
y\left(1,1\right) - y\left(0,1\right) \in 
\begin{cases}
\left[\underline{K} - y\left(z\right), \overline{K} - \underline{K}\right] & \text{if } z = \left(0,0\right) \\
\left[y\left(z\right) - \overline{K}, \overline{K} - \underline{K}\right] & \text{if } z = \left(1,0\right) \\
\left[\underline{K} - y\left(z\right), \overline{K} - y\left(z\right)\right] & \text{if } z = \left(0,1\right) \\
\left[y\left(z\right) - \overline{K}, y\left(z\right) - \underline{K}\right] & \text{if } z = \left(1,1\right) \\
\end{cases}
\end{equation}

\begin{equation}\label{spm zero 3 proof}
y\left(1,1\right) - y\left(0,0\right) \in 
\begin{cases}
\left[\underline{K} - y\left(z\right), \overline{K} - y\left(z\right)\right] & \text{if } z = \left(0,0\right) \\
\left[\underline{K} - \overline{K}, \overline{K} - \underline{K}\right] & \text{if } z = \left(1,0\right) \\
\left[\underline{K} - \overline{K}, \overline{K} - \underline{K}\right] & \text{if } z = \left(0,1\right) \\
\left[y\left(z\right) - \overline{K}, y\left(z\right) - \underline{K}\right] & \text{if } z = \left(1,1\right) \\
\end{cases}
\end{equation}

Taking expectations in equations \eqref{spm zero 1 proof} and \eqref{spm zero 2 proof} yields the bounds in \eqref{spm zero 1} and \eqref{spm zero 2}, respectively. Equation \eqref{spm zero 3 proof} shows that the no--assumption bounds remain sharp for $\E\left[\, y\left(1,1\right) - y\left(0,0\right) \,\right]$. A similar analysis establishes that the no--assumption bounds remain sharp for $\E\left[\, y\left(0,1\right) - y\left(1,0\right) \,\right]$ and $\E\left[\, y\left(1,0\right) - y\left(0,1\right) \,\right]$.

\end{proof}

In proposition \ref{spm bounds 0}, assumption \ref{SPM} improves the upper bound on $\E\left[\, y\left(1,0\right) - y\left(0,0\right) \,\right]$ and the lower bound on $\E\left[\, y\left(1,1\right) - y\left(0,1\right) \,\right]$ by establishing a monotonicity relationship between the two treatment effects. Similar bounds for the treatment effects $\E\left[\, y\left(0,1\right) - y\left(0,0\right) \,\right]$ and $\E\left[\, y\left(1,1\right) - y\left(1,0\right) \,\right]$ can be obtained by permuting the order of the treatments and applying the result. As stated in the proposition, the no--assumption bounds on $\E\left[\, y\left(1,1\right) - y\left(0,0\right) \,\right]$, $\E\left[\, y\left(1,0\right) - y\left(0,1\right) \,\right]$, and $\E\left[\, y\left(0,1\right) - y\left(1,0\right) \,\right]$ remain sharp under \ref{SPM}. In the special case where $y$ is bounded between zero and one, \ref{SPM} can establish that $\E\left[\, y\left(1,0\right) - y\left(0,0\right) \,\right]\in \left[-1, 0\right]$ or $\E\left[\, y\left(1,1\right) - y\left(0,1\right) \,\right]\in\left[0,1\right]$ if the observed expectations in \eqref{spm zero 1} and \eqref{spm zero 2} take certain boundary values. In general, however, \ref{SPM} is not sufficient to identify the sign of a treatment effect in the absence of other assumptions.

Sharp bounds can be derived on more general treatment spaces using the same approach. For any $t_{1}$ and $t_{2}$, 
\begin{equation*}
\E\left[\, y\left(t_{1}\right) - y\left(t_{2}\right) \,\right] = \sum_{t_{3}\in T} \E\left[\, y\left(t_{1}\right) - y\left(t_{2}\right) \mid z = t_{3}\,\right] P\left(z=t_{3}\right)
\end{equation*}
Thus, bounds can be derived for $\E\left[\, y\left(t_{1}\right) - y\left(t_{2}\right) \,\right]$ by bounding $\E\left[\, y\left(t_{1}\right) - y\left(t_{2}\right) \mid z=t_{3}\,\right]$ for each $t_{3}$. Bounds on $\E\left[\, y\left(t_{1}\right) - y\left(t_{2}\right) \mid z=t_{3}\,\right]$ will differ depending on the data, the ordering of the three treatments, and whether \ref{SPM} and/or \ref{SBM} hold on any sublattices containing the treatments. To capture all of the information implied by the data and assumptions, one must examine every sublattice containing $t_{1}$, $t_{2}$, and $t_{3}$, determine which of \ref{SPM} and \ref{SBM} hold on the sublattice, and derive the implied bounds given the ordering of $t_{1}$, $t_{2}$, and $t_{3}$. 

To this end, I formally define a collection of sets of treatments (indexed by $t_{1}$ and $t_{2}$) which will allow for a derivation of the bounds on $\E\left[\, y\left(t_{1}\right) - y\left(t_{2}\right) \mid z=t_{3}\,\right]$ by identifying which assumptions (if any) imply tightened bounds and the form those bounds will take. Let $\left\{S_{\gamma}\right\}_{\gamma\in\Gamma}$ be the collection of all sublattices of $T$ such that, for every $\gamma\in\Gamma$, $\left|S_{\gamma}\right| = 4$ and $S_{\gamma}$ is not a chain. These are the minimal sublattices on which \ref{SPM} and \ref{SBM} can have any implications for identification. Let $\Gamma_{t_{1},t_{2}}$ denote the set of $\gamma$ such that $t_{1},t_{2}\in S_{\gamma}$, and let $\Gamma_{t_{1},t_{2}}^{\ref{SPM}}$ refer to the subset of $\Gamma_{t_{1},t_{2}}$ for which \ref{SPM} (but not \ref{SBM}) holds on $S_{\gamma}$; likewise, let $\Gamma_{t_{1},t_{2}}^{\ref{SBM}}$ refer to the subset where only \ref{SBM} holds on $S_{\gamma}$. Let $\Gamma_{t_{1},t_{2}}^{MOD}$ refer to the subset where both \ref{SPM} and \ref{SBM} hold. Then, one can define the following sets:
\begin{equation}\label{lambdas}
\begin{gathered}
\Lambda_1 = 
\left\{t_{3} \mid t_{2} < t_{1} < t_{3} \text{ and either } \exists \gamma\in\Gamma_{t_{1},t_{2}}^{MOD} \text{ s.t. } t_{3}\in S_{\gamma} \right. \\
\left. \text{ or } \exists \gamma\in\Gamma_{t_{1},t_{2}}^{\ref{SPM}}, \gamma^{\p}\in\Gamma_{t_{1},t_{2}}^{\ref{SBM}} \text{ s.t. } t_{3}\in S_{\gamma}\cap S_{\gamma^{\p}} \right\} \\
\Lambda_2 = 
\left\{t_{3} \mid t_{2} < t_{1} < t_{3}, \nexists \gamma\in\Gamma_{t_{1},t_{2}}^{MOD}\cup\Gamma_{t_{1},t_{2}}^{\ref{SBM}} \text{ s.t. } t_{3}\in S_{\gamma}, \right. \\
\left. \text{ and } \exists \gamma\in\Gamma_{t_{1},t_{2}}^{\ref{SPM}} \text{ s.t. } t_{3}\in S_{\gamma} \right\} \\
\Lambda_3 = 
\left\{t_{3} \mid t_{2} < t_{1} < t_{3}, \nexists \gamma\in\Gamma_{t_{1},t_{2}}^{MOD}\cup\Gamma_{t_{1},t_{2}}^{\ref{SPM}} \text{ s.t. } t_{3}\in S_{\gamma}, \right. \\
\left. \text{ and } \exists \gamma\in\Gamma_{t_{1},t_{2}}^{\ref{SBM}} \text{ s.t. } t_{3}\in S_{\gamma} \right\} \\
\Lambda_4 = 
\left\{t_{3} \mid t_{3} < t_{2} < t_{1} \text{ and either } \exists \gamma\in\Gamma_{t_{1},t_{2}}^{MOD} \text{ s.t. } t_{3}\in S_{\gamma} \right. \\
\left. \text{ or } \exists \gamma\in\Gamma_{t_{1},t_{2}}^{\ref{SPM}}, \gamma^{\p}\in\Gamma_{t_{1},t_{2}}^{\ref{SBM}} \text{ s.t. } t_{3}\in S_{\gamma}\cap S_{\gamma^{\p}} \right\} \\
\Lambda_5 = 
\left\{t_{3} \mid t_{3} < t_{2} < t_{1}, \nexists \gamma\in\Gamma_{t_{1},t_{2}}^{MOD}\cup\Gamma_{t_{1},t_{2}}^{\ref{SBM}} \text{ s.t. } t_{3}\in S_{\gamma}, \right. \\
\left. \text{ and } \exists \gamma\in\Gamma_{t_{1},t_{2}}^{\ref{SPM}} \text{ s.t. } t_{3}\in S_{\gamma} \right\} \\
\Lambda_6 = 
\left\{t_{3} \mid t_{3} < t_{2} < t_{1}, \nexists \gamma\in\Gamma_{t_{1},t_{2}}^{MOD}\cup\Gamma_{t_{1},t_{2}}^{\ref{SPM}} \text{ s.t. } t_{3}\in S_{\gamma}, \right. \\
\left. \text{ and } \exists \gamma\in\Gamma_{t_{1},t_{2}}^{\ref{SBM}} \text{ s.t. } t_{3}\in S_{\gamma} \right\} \\
\Lambda_7 = 
\left\{t_{3} \mid t_{3} < t_{1}, t_{3} \parallel t_{2}, \text{ and } \exists \gamma\in\Gamma_{t_{1},t_{2}}^{MOD} \text{ s.t. } t_{3}\in S_{\gamma} \right\} \\
\Lambda_8 = 
\left\{t_{3} \mid t_{3} < t_{1}, t_{3} \parallel t_{2}, \text{ and } \exists \gamma\in\Gamma_{t_{1},t_{2}}^{\ref{SPM}} \text{ s.t. } t_{3}\in S_{\gamma} \right\} \\
\Lambda_9 = 
\left\{t_{3} \mid t_{3} < t_{1}, t_{3} \parallel t_{2}, \text{ and } \exists \gamma\in\Gamma_{t_{1},t_{2}}^{\ref{SBM}} \text{ s.t. } t_{3}\in S_{\gamma} \right\} \\
\Lambda_{10} = 
\left\{t_{3} \mid t_{2} < t_{3}, t_{3} \parallel t_{1}, \text{ and } \exists \gamma\in\Gamma_{t_{1},t_{2}}^{MOD} \text{ s.t. } t_{3}\in S_{\gamma} \right\} \\
\Lambda_{11} = 
\left\{t_{3} \mid t_{2} < t_{3}, t_{3} \parallel t_{1}, \text{ and } \exists \gamma\in\Gamma_{t_{1},t_{2}}^{\ref{SPM}} \text{ s.t. } t_{3}\in S_{\gamma} \right\} \\
\Lambda_{12} = 
\left\{t_{3} \mid t_{2} < t_{3}, t_{3} \parallel t_{1}, \text{ and } \exists \gamma\in\Gamma_{t_{1},t_{2}}^{\ref{SBM}} \text{ s.t. } t_{3}\in S_{\gamma} \right\} \\
\Lambda_{13} = \left(\bigcup_{j=1}^{12}\Lambda_{j}\right)^{c}\setminus\left\{t_{1}, t_{2}\right\} 
\end{gathered}
\end{equation}
These partition the set of sublattices of $T$ containing $t_{1}$, $t_{2}$, and $t_{3}$ based on the ordering of the treatments and the assumptions imposed. These sets are mutually exclusive and exhaustive, so each $t_{3}$ will belong to exactly one. For an example of how the sets can be used, consider $\Lambda_{5}$. If $t_3\in \Lambda_5$, then $t_{1}$, $t_{2}$, and $t_{3}$ belong to one or more sublattices on which \ref{SPM} holds, but none on which \ref{SBM} holds. The required ordering implies that 
\[
\underline{K} - y\left(t_{3}\right) \leq y\left(t_{1}\right) - y\left(t_{2}\right)
\] 
and thus
\begin{equation*}
y\left(t_{1}\right) - y\left(t_{2}\right) \in \left[\underline{K} - y\left(t_{3}\right), \overline{K} - \underline{K} \right]
\end{equation*}
which implies that 
\[
\underline{K} - \E\left[\, y\left(t_{3}\right) \mid z = t_{3} \,\right] \\
\leq \E\left[\, y\left(t_{1}\right) - y\left(t_{2}\right) \mid z = t_{3}\,\right] \leq \\
\overline{K} - \underline{K}
\]

The following result uses the above sets to derive sharp bounds on $\E\left[\, y\left(t_{1}\right) - y\left(t_{2}\right) \,\right]$ under arbitrary combinations of supermodularity and submodularity.

\begin{proposition}\label{spm bounds 2}
Let $\left\{S_{\gamma}\right\}_{\gamma\in\Gamma}$ be the collection of all sublattices of $T$ such that, for every $\gamma\in\Gamma$, $S_{\gamma}$ is not a chain and $\left|S_{\gamma}\right| = 4$. Define $\Gamma^{\ref{SPM}}\subseteq\Gamma$ to be the set of $\gamma$ such that \ref{SPM} holds on $S_{\gamma}$ and \ref{SBM} does not hold on $S_{\gamma}$ iff $\gamma \in \Gamma^{\ref{SPM}}$; likewise, define $\Gamma^{\ref{SBM}}\subseteq\Gamma$ to be the set of $\gamma$ such that \ref{SBM} holds on $S_{\gamma}$ and \ref{SPM} does not hold on $S_{\gamma}$ iff $\gamma \in \Gamma^{\ref{SBM}}$. Define $\Gamma^{MOD}\subseteq\Gamma$ to be the set of $\gamma$ such that \ref{SPM} and \ref{SBM} hold on $S_{\gamma}$ iff $\gamma \in \Gamma^{MOD}$. Let $\Gamma_{t_{1},t_{2}}^{\ref{SPM}}\subseteq \Gamma^{\ref{SPM}}$ be the set of $\gamma$ such that $t_{1},t_{2}\in S_{\gamma}$ and $\gamma\in\Gamma^{\ref{SPM}}$; likewise for $\Gamma_{t_{1},t_{2}}^{\ref{SBM}}$ and $\Gamma_{t_{1},t_{2}}^{MOD}$. Then, for $t_{2} < t_{1}$,
\begin{equation}\label{sharp bounds 2}
\begin{gathered}
\left[\E\left[\, y\left(t_{1}\right) \mid z=t_{1} \,\right] - \overline{K}\right]P\left(z=t_{1}\right) + 
\left[\underline{K} - \E\left[\, y\left(t_{2}\right) \mid z=t_{2} \,\right]\right]P\left(z=t_{2}\right) \\
+ \sum_{t_{3}\in\bigcup_{j\in\left\{1,3,7,8\right\}}\Lambda_{j}}\left[\E\left[\, y\left(t_{3}\right) \mid z=t_{3} \,\right] - \overline{K}\right]P\left(z=t_{3}\right) \\
+ \sum_{t_{3}\in\bigcup_{j\in\left\{4,5,10,12\right\}}\Lambda_{j}}\left[\underline{K} - \E\left[\, y\left(t_{3}\right) \mid z=t_{3} \,\right]\right]P\left(z=t_{3}\right) \\
+ \sum_{t_{3}\in\bigcup_{j\in\left\{2,6,9,11,13\right\}}\Lambda_{j}} \left[\underline{K} - \overline{K}\right]P\left(z=t_{3}\right) \\
\leq \E\left[\, y\left(t_{1}\right) - y\left(t_{2}\right) \,\right] \leq \\
\left[\overline{K} - \E\left[\, y\left(t_{2}\right) \mid z=t_{2} \,\right]\right]P\left(z=t_{2}\right) 
+ \left[\E\left[\, y\left(t_{1}\right) \mid z=t_{1} \,\right] - \underline{K}\right]P\left(z=t_{1}\right) \\
+ \sum_{t_{3}\in\bigcup_{j\in\left\{1,2,7,9\right\}}\Lambda_{j}}\left[\E\left[\, y\left(t_{3}\right) \mid z=t_{3} \,\right] - \underline{K}\right]P\left(z=t_{3}\right) \\
+ \sum_{t_{3}\in\bigcup_{j\in\left\{4,6,10,11\right\}}\Lambda_{j}}\left[\overline{K} - \E\left[\, y\left(t_{3}\right) \mid z=t_{3} \,\right]\right]P\left(z=t_{3}\right) \\
+ \sum_{t_{3}\in\bigcup_{j\in\left\{3,5,8,12,13\right\}}\Lambda_{j}} \left[\overline{K} - \underline{K}\right]P\left(z=t_{3}\right) \\
\end{gathered}
\end{equation}
These bounds are sharp.
\end{proposition}

\begin{proof}[Proof of proposition \ref{spm bounds 2}]

By the Law of Iterated Expectations, 
\begin{equation}\label{lie}
\E\left[\, y\left(t_{1}\right) - y\left(t_{2}\right) \,\right] = \sum_{t_{3}\in T} \E\left[\, y\left(t_{1}\right) - y\left(t_{2}\right) \mid z = t_{3}\,\right] P\left(z=t_{3}\right)
\end{equation}

Sharp bounds for the partially identified expectations on the right hand side of \eqref{lie} will yield sharp bounds on $\E\left[\, y\left(t_{1}\right) - y\left(t_{2}\right) \,\right]$. I proceed by finding the sharp identification region for an arbitrary $y\left(t_{1}\right) - y\left(t_{2}\right)$ and every possible value that $z$ may take. These can be averaged to find sharp bounds on $\E\left[\, y\left(t_{1}\right) - y\left(t_{2}\right) \mid z \,\right]$ for all $z$. \ref{SPM} and \ref{SBM} provide identifying power by establishing monotonicity relationships between treatment effects, so for each possible value of $z$, it must be determined which of the following inequalities are implied by the maintained assumptions:
\begin{gather}
y\left(t_{1}\right) - y\left(t_{2}\right) \leq y\left(z\right) - \underline{K} \label{rightleft} \\
y\left(t_{1}\right) - y\left(t_{2}\right) \leq \overline{K} - y\left(z\right) \label{rightright}\\
y\left(z\right) - \overline{K} \leq y\left(t_{1}\right) - y\left(t_{2}\right) \label{leftleft}\\
\underline{K} - y\left(z\right) \leq y\left(t_{1}\right) - y\left(t_{2}\right) \label{leftright}
\end{gather} 

I show that \begin{inparaenum}[\itshape 1)] \item the union of these sets is $T\setminus\left\{t_{1}, t_{2}\right\}$, \item these sets are mutually exclusive, and \item they allow for a characterization of the identification region for $\E\left[\, y\left(t_{1}\right) - y\left(t_{2}\right) \,\right]$\end{inparaenum}. $\bigcup_{j=1}^{12}\Lambda_{j} = T\setminus\left\{t_{1}, t_{2}\right\}$ follows from the definition of $\Lambda_{13}$. Due to the orderings required for membership, the sets $\Lambda_{1}\cup\Lambda_{2}\cup\Lambda_{3}$, $\Lambda_{4}\cup\Lambda_{5}\cup\Lambda_{6}$, $\Lambda_{7}\cup\Lambda_{8}\cup\Lambda_{9}$, and $\Lambda_{10}\cup\Lambda_{11}\cup\Lambda_{12}$ are mutually exclusive. 

The sets $\Lambda_{1}$, $\Lambda_{2}$, and $\Lambda_{3}$ collect those sublattices where $t_{2}$ is the bottom and $t_{1}$ is in the interior. $\Lambda_{1}$, $\Lambda_{2}$, and $\Lambda_{3}$ are mutually exclusive because $\Lambda_{1}$ requires either the existence of a four--point sublattice containing $t_{1}$, $t_{2}$, and $t_{3}$ on which both \ref{SPM} and \ref{SBM} hold or the existence of at least two distinct sublattices containing $t_{1}$, $t_{2}$, and $t_{3}$ where \ref{SPM} holds on one while \ref{SBM} holds on the other. $\Lambda_2$ rules out both of these possibilities, only allowing for the existence of a sublattice containing $t_{1}$, $t_{2}$, and $t_{3}$ on which \ref{SPM} holds; similarly, $\Lambda_3$ only allows for the existence of a sublattice containing $t_{1}$, $t_{2}$, and $t_{3}$ on which \ref{SBM} holds. An identical argument establishes that $\Lambda_{4}$, $\Lambda_{5}$, and $\Lambda_{6}$ are mutually exclusive; these sets are defined identically to $\Lambda_{1}$, $\Lambda_{2}$, and $\Lambda_{3}$, respectively, with the modification that $t_{3} < t_{2} < t_{1}$ rather than $t_{2} < t_{1} < t_{3}$, so $t_{2}$ is in the interior and $t_{1}$ is the top of each sublattice. 

The set construction is complicated in the above cases because the orderings $t_{2} < t_{1} < t_{3}$ and $t_{3} < t_{2} < t_{1}$ are compatible with multiple four--point sublattices containing $t_{1}$, $t_{2}$, and $t_{3}$. This results from the fact that there may be multiple $t_{4}$ such that $t_{1}\vee t_{4} = t_{3}$ and $t_{1}\wedge t_{4} = t_{2}$ (in the former case) and $t_{2}\vee t_{4} = t_{1}$ and $t_{2}\wedge t_{4} = t_{3}$ (in the latter case). When exactly two of $t_{1}$, $t_{2}$, and $t_{3}$ are incomparable, there can be at most one four--point sublattice containing $t_{1}$, $t_{2}$, and $t_{3}$, since the incomparable treatments define a unique meet and join. This simplifies the construction of the sets $\Lambda_{7},\ldots,\Lambda_{12}$. By the above argument and the fact that
\[
\Gamma_{t_{1},t_{2}}^{MOD}\cap\Gamma_{t_{1},t_{2}}^{\ref{SPM}} = \Gamma_{t_{1},t_{2}}^{MOD}\cap\Gamma_{t_{1},t_{2}}^{\ref{SBM}} = \Gamma_{t_{1},t_{2}}^{\ref{SPM}}\cap\Gamma_{t_{1},t_{2}}^{\ref{SBM}} = \emptyset
\]
the sets $\Lambda_{7}$, $\Lambda_{8}$, and $\Lambda_{9}$ are mutually exclusive. The same argument establishes that $\Lambda_{10}$, $\Lambda_{11}$, and $\Lambda_{12}$ are mutually exclusive.

The sets $\Lambda_{1},\ldots,\Lambda_{12}$ define every sublattice membership pattern for $t_{1}$ and $t_{2}$ for which \ref{SPM} and \ref{SBM} may have any implications; this follows from proposition \ref{spm bounds 0} and its straightforward extension to the case of \ref{SBM}. I now outline the implications for the identification of $y\left(t_{1}\right) - y\left(t_{2}\right)$ when $z\in \Lambda_j$ for each $j\in\left\{1,\ldots,12\right\}$.

If $z\in \Lambda_1$, then both \ref{SPM} and \ref{SBM} hold on one or more sublattices containing $t_{1}$, $t_{2}$, and $z$, and the required ordering implies that both \eqref{rightleft} and \eqref{leftleft} hold. Thus, 
\begin{equation}\label{lambda1}
y\left(t_{1}\right) - y\left(t_{2}\right) \in \left[y\left(z\right) - \overline{K}, y\left(z\right) - \underline{K}\right]
\end{equation}
If $z\in \Lambda_2$, then $t_{1}$, $t_{2}$, and $z$ belong to one or more sublattices on which only \ref{SPM} holds. The required ordering implies that \eqref{leftright} holds and thus
\begin{equation}\label{lambda2}
y\left(t_{1}\right) - y\left(t_{2}\right) \in \left[\underline{K} - \overline{K}, y\left(z\right) - \underline{K}\right]
\end{equation}
If $z\in \Lambda_3$, then $t_{1}$, $t_{2}$, and $z$ belong to one or more sublattices on which only \ref{SBM} holds. The required ordering implies that \eqref{leftleft} holds and thus
\begin{equation}\label{lambda3}
y\left(t_{1}\right) - y\left(t_{2}\right) \in \left[y\left(z\right) - \overline{K}, \overline{K} - \underline{K}\right]
\end{equation}
If $z\in \Lambda_4$, then both \ref{SPM} and \ref{SBM} hold on one or more sublattices containing $t_{1}$, $t_{2}$, and $z$, and the required ordering implies that both \eqref{rightright} and \eqref{leftright} hold. Thus, 
\begin{equation}\label{lambda4}
y\left(t_{1}\right) - y\left(t_{2}\right) \in \left[\underline{K} - y\left(z\right), \overline{K} - y\left(z\right) \right]
\end{equation}
If $z\in \Lambda_5$, then $t_{1}$, $t_{2}$, and $z$ belong to one or more sublattices on which only \ref{SPM} holds. The required ordering implies that \eqref{leftright} holds and thus
\begin{equation}\label{lambda5}
y\left(t_{1}\right) - y\left(t_{2}\right) \in \left[\underline{K} - y\left(z\right), \overline{K} - \underline{K} \right]
\end{equation}
If $z\in \Lambda_6$, then $t_{1}$, $t_{2}$, and $z$ belong to one or more sublattices on which only \ref{SBM} holds. The required ordering implies that \eqref{rightright} holds and thus
\begin{equation}\label{lambda6}
y\left(t_{1}\right) - y\left(t_{2}\right) \in \left[\underline{K} - \overline{K}, \overline{K} - y\left(z\right) \right]
\end{equation}
If $z\in \Lambda_7$, then both \ref{SPM} and \ref{SBM} hold on the only sublattice containing $t_{1}$, $t_{2}$, and $z$, and the required ordering implies that both \eqref{rightleft} and \eqref{leftleft} hold. Thus, 
\begin{equation}\label{lambda7}
y\left(t_{1}\right) - y\left(t_{2}\right) \in \left[y\left(z\right) - \overline{K}, y\left(z\right) - \underline{K}\right]
\end{equation}
If $z\in \Lambda_8$, then \ref{SPM} holds on the only sublattice containing $t_{1}$, $t_{2}$, and $z$, and the required ordering implies that \eqref{leftleft} holds and thus
\begin{equation}\label{lambda8}
y\left(t_{1}\right) - y\left(t_{2}\right) \in \left[y\left(z\right) - \overline{K}, \overline{K} - \underline{K}\right]
\end{equation}
If $z\in \Lambda_9$, then \ref{SBM} holds on the only sublattice containing $t_{1}$, $t_{2}$, and $z$, and the required ordering implies that \eqref{rightleft} holds and thus
\begin{equation}\label{lambda9}
y\left(t_{1}\right) - y\left(t_{2}\right) \in \left[\underline{K} - \overline{K}, y\left(z\right) - \underline{K}\right]
\end{equation}
If $z\in \Lambda_{10}$, then both \ref{SPM} and \ref{SBM} hold on the only sublattice containing $t_{1}$, $t_{2}$, and $z$, and the required ordering implies that both \eqref{leftright} and \eqref{rightright} hold. Thus, 
\begin{equation}\label{lambda10}
y\left(t_{1}\right) - y\left(t_{2}\right) \in \left[\underline{K} - y\left(z\right), \overline{K} - y\left(z\right)\right]
\end{equation}
If $z\in \Lambda_{11}$, then \ref{SPM} holds on the only sublattice containing $t_{1}$, $t_{2}$, and $z$, and the required ordering implies that \eqref{rightright} holds and thus
\begin{equation}\label{lambda11}
y\left(t_{1}\right) - y\left(t_{2}\right) \in \left[\underline{K} - \overline{K}, \overline{K} - y\left(z\right) \right]
\end{equation}
If $z\in \Lambda_{12}$, then \ref{SBM} holds on the only sublattice containing $t_{1}$, $t_{2}$, and $z$, and the required ordering implies that \eqref{leftright} holds and thus
\begin{equation}\label{lambda12}
y\left(t_{1}\right) - y\left(t_{2}\right) \in \left[\underline{K} - y\left(z\right), \overline{K} - \underline{K} \right]
\end{equation}

The set $\Lambda_{13}$ contains those $t_{3}$ that obey one of the orderings from $\Lambda_{1},\ldots,\Lambda_{12}$ but do not belong to a sublattice on which \ref{SPM} or \ref{SBM} hold; thus, the no--assumption bounds $\left[\underline{K} - \overline{K}, \overline{K} - \underline{K}\right]$ will hold for these. It also contains $t_{3}$ such that $t_{2} < t_{3} < t_{1}$, and any four--point sublattice containing these treatments must have $t_{2}$ as the bottom and $t_{1}$ as the top; on these sublattices, \ref{SPM} and \ref{SBM} have no implications for identification. Additionally, $\Lambda_{13}$ contains $t_{3}$ such that $t_{3} \parallel t_{1}$ and $t_{3} \parallel t_{2}$; \ref{SPM} and \ref{SBM} have no identifying power in these circumstances, as $t_{1}$, $t_{2}$, and $t_{3}$ do not share a four--point sublattice.

The focus on four--point sublattices is without loss of generality, since the implications of assumptions \ref{SPM} and \ref{SBM} only appear on four--point sublattices. \ref{SPM} and \ref{SBM} have no implications on chains, so sublattices that are chains can be ignored. Restricting attention to elements of $\left\{S_{\gamma}\right\}_{\gamma\in\Gamma_{t_{1},t_{2}}}\subseteq\left\{S_{\gamma}\right\}_{\gamma\in\Gamma}$ is without loss of generality as well. This follows from the fact that \ref{SPM} and \ref{SBM} have no implications for potential outcomes under the maintained assumptions, and any implications for the treatment effect $y\left(t_{1}\right) - y\left(t_{2}\right)$ from another treatment effect which are mediated by a third treatment effect are realized directly on a sublattice containing the treatments from the first two treatment effects. To see this concretely, suppose that $S_{\gamma} = \left\{t_{2}, t_{1}, t_{3}, t_{4}\right\}$ and $S_{\gamma^{\p}} = \left\{t_{3}, t_{4}, t_{5}, t_{6}\right\}$ where $t_{1}\parallel t_{3}$, $t_{4}\parallel t_{5}$, $t_{2} = t_{1}\wedge t_{3}$, $t_{4} = t_{1} \vee t_{3}$, $t_{3} = t_{4}\wedge t_{5}$, and $t_{6} = t_{4}\vee t_{5}$. Suppose that \ref{SPM} holds on both $S_{\gamma}$ and $S_{\gamma^{\p}}$. This implies
\begin{equation*}
\begin{gathered}
y\left(t_{1}\right) - y\left(t_{2}\right) \leq y\left(t_{4}\right) - y\left(t_{3}\right) \leq y\left(t_{6}\right) - y\left(t_{5}\right) \\
\Longrightarrow y\left(t_{1}\right) - y\left(t_{2}\right) \leq y\left(t_{6}\right) - y\left(t_{5}\right)
\end{gathered}
\end{equation*}
The fact that $\left\{t_{2},t_{1},t_{5},t_{6}\right\} \in \left\{S_{\gamma}\right\}_{\gamma\in\Gamma^{\ref{SPM}}_{t_{1},t_{2}}}$ follows from lemma \ref{lattice lemma} below and the definition of $\Gamma^{\ref{SPM}}_{t_{2},t_{1}}$:

\begin{lemma}\label{lattice lemma}
Assume that $t_{1}\parallel t_{3}$, $t_{4}\parallel t_{5}$, $t_{2} = t_{1}\wedge t_{3}$, $t_{4} = t_{1} \vee t_{3}$, $t_{3} = t_{4}\wedge t_{5}$, and $t_{6} = t_{4}\vee t_{5}$. Then, $t_{5}\wedge t_{1} = t_{2}$ and $t_{5}\vee t_{1} = t_{6}$.
\end{lemma}
\begin{proof}
See appendix.
\end{proof}

The argument in the \ref{SBM} case simply reverses the inequalities.

The results from \eqref{lambda1}--\eqref{lambda12} are summarized below:

\begin{equation}\label{spm bounds 2 proof 1}
y\left(t_{1}\right) - y\left(t_{2}\right) \in 
\begin{cases}
\left[\underline{K} - y\left(z\right), \overline{K} - y\left(z\right)\right] & \text{if } z \in \left\{t_{2}\right\}\cup \Lambda_4 \cup \Lambda_{10} \\
\left[y\left(z\right) - \overline{K}, y\left(z\right) - \underline{K}\right] & \text{if } z \in \left\{t_{1}\right\} \cup \Lambda_1 \cup \Lambda_{7} \\
\left[\underline{K} - \overline{K}, y\left(z\right) - \underline{K}\right] & \text{if } z \in \Lambda_2 \cup \Lambda_{9} \\
\left[\underline{K} - \overline{K}, \overline{K} - y\left(z\right)\right] & \text{if } z \in \Lambda_6 \cup \Lambda_{11} \\
\left[y\left(z\right) - \overline{K}, \overline{K} - \underline{K}\right] & \text{if } z \in \Lambda_3 \cup \Lambda_{8} \\
\left[\underline{K} - y\left(z\right), \overline{K} - \underline{K}\right] & \text{if } z \in \Lambda_5 \cup \Lambda_{12} \\
\left[\underline{K} - \overline{K}, \overline{K} - \underline{K}\right] & \text{if } z \in \Lambda_{13} \\
\end{cases}
\end{equation}

Since the sets $\left\{t_{1}\right\},\left\{t_{2}\right\},\Lambda_{1},\ldots,\Lambda_{13}$ are mutually exclusive and exhaustive, \eqref{spm bounds 2 proof 1} characterizes the sharp identification region for $y\left(t_{1}\right) - y\left(t_{2}\right)$ and each possible $z$. Averaging the bounds in \eqref{spm bounds 2 proof 1} across $i$ yields sharp bounds on $\E\left[\, y\left(t_{1}\right) - y\left(t_{2}\right) \,\right]$ via \eqref{lie}. These are given in \eqref{sharp bounds 2}.

\end{proof}

Proposition \ref{spm bounds 2} generalizes proposition \ref{spm bounds 0} by allowing for a much richer set of treatments. The treatment may have any finite number of dimensions, and each may be binary or multivalued. Some dimensions of the treatment may be complements while others are substitutes; the result allows for arbitrary combinations of \ref{SPM} and \ref{SBM} as appropriate. The complexity of the result is due to two factors. First, the treatment pair $t_{1}, t_{2}$ may belong to multiple sublattices. Second, the position of the treatment pair within a lattice, i.e., whether it includes the top and/or bottom of the sublattice, differs across sublattices. The position of the treatment pair within a sublattice combined with the assumptions that hold on the sublattice determine whether the upper and/or lower bound (or neither) are improved. While the result appears complicated, defining the $\Lambda$ sets in practice seems to be fairly straightforward.

I have focused on bounding expectations of treatment effects using only supermodularity and submodularity assumptions. However, in applications these will often be paired with other assumptions, such as monotonicity. The next result modifies proposition \ref{spm bounds 2} by adding Manski's assumption of semi--monotone treatment response:

\begin{proposition}\label{spm bounds 3}
Assume the conditions of proposition \ref{spm bounds 2} and that SMTR holds on $T$. Then, for $t_{2} < t_{1}$,
\begin{equation}\label{sharp bounds 3}
\begin{gathered}
0 \leq \E\left[\, y\left(t_{1}\right) - y\left(t_{2}\right) \,\right] \leq \\
\left[\overline{K} - \E\left[\, y\left(t_{2}\right) \mid z\leq t_{2} \,\right]\right]P\left(z \leq t_{2}\right) 
+ \left[\E\left[\, y\left(t_{1}\right) \mid z \geq t_{1} \,\right] - \underline{K}\right]P\left(z\geq t_{1}\right) \\
+ \sum_{t_{3}\in\bigcup_{j\in\left\{7,9\right\}}\Lambda_{j}}\left[\E\left[\, y\left(t_{3}\right) \mid z=t_{3} \,\right] - \underline{K}\right]P\left(z=t_{3}\right) \\
+ \sum_{t_{3}\in\cup\bigcup_{j\in\left\{10,11\right\}}\Lambda_{j}}\left[\overline{K} - \E\left[\, y\left(t_{3}\right) \mid z=t_{3} \,\right]\right]P\left(z=t_{3}\right) \\
+ \sum_{t_{3}\in\left\{t_{3} \mid t_{1}\nleq t_{3} \text{ and } t_{3}\nleq t_{2}\right\} \cap\bigcup_{j\in\left\{8,12,13\right\}}\Lambda_{j}} \left[\overline{K} - \underline{K}\right]P\left(z=t_{3}\right) \\
\end{gathered}
\end{equation}
These bounds are sharp.
\end{proposition}

Proposition \ref{spm bounds 3} shows that bounds computed under combinations of SMTR and \ref{SPM}/\ref{SBM} will generally be narrower than bounds computed under the assumptions separately. In particular, \ref{SPM} and \ref{SBM} may tighten the SMTR upper bound, and SMTR at least weakly improves the \ref{SPM}/\ref{SBM} lower bound and may tighten the upper bound as well. \eqref{spm bounds 3 proof 1} shows that when $z\in \Lambda_{7} \cup \Lambda_{9}$, \ref{SBM} will improve the SMTR upper bound, and when $z\in \Lambda_{10} \cup \Lambda_{11}$, \ref{SPM} will improve the SMTR upper bound. The SMTR lower bound at zero cannot be improved by using \ref{SPM} or \ref{SBM}. In the empirical application in section \ref{emp}, I show how the addition of \ref{SBM} improves the SMTR bounds with real data.  

\end{section}

\begin{section}{Instrumental Variables}\label{instrumental variables}

Traditional instrumental variable (IV) analysis of treatment response relies on the existence of a variable that is correlated with the treatment variable of interest but is mean--independent or independent of the distribution of response functions. Whether or not such independence assumptions are justified in a particular context is often the subject of vigorous debate. This has motivated researchers to find weaker and more credible forms of these assumptions that still retain some identification power. A leading example is the notion of a monotone instrumental variable (MIV): A variable $x$ is an MIV if average potential outcomes conditional on $x$ are monotone in $x$ \citep{manski2000monotone, manski2009more}. MIV and its generalizations impose restrictions on functionals of potential outcome distributions. Restrictions can also be imposed directly on functionals of treatment effect distributions:

\begin{assumptionSPMIV}[Supermodular instrumental variable]\label{SPMIV}
$x^k$ is a \textit{supermodular instrumental variable} for $\E\left[\, y\left(t_{1}\right) - y\left(t_{2}\right) \mid x^k,x^{-k} \,\right]$ with $t_{2}\leq t_{1}$ if
\begin{gather}\label{spmiv 0}
x_{1}^k \leq x^{k}_{2}  
\Longrightarrow \E\left[\, y\left(t_{1}\right) - y\left(t_{2}\right) \mid x_{1}^k,x^{-k} \,\right] \leq \E\left[\, y\left(t_{1}\right) - y\left(t_{2}\right) \mid x^{k}_{2},x^{-k} \,\right]
\end{gather}
for all $x^{-k}$.\footnote{The weak inequality in \eqref{spmiv 0} can be reversed, in which case $x^k$ would be a \textit{submodular instrumental variable} (SBMIV). If the inequality is replaced with equality, $x^k$ becomes a \textit{modular instrumental variable} (MODIV).}
\end{assumptionSPMIV}

\ref{SPMIV} is an alternative formulation of complementarity where average treatment effects vary monotonically with an observed covariate $x^k$.\footnote{The \ref{SPM}/\ref{SPMIV} distinction is analogous to the MTR/MIV distinction; see \cite{manski2009more}.} An advantage of these assumptions is that evidence for their validity may be provided by previous studies where strong identifying assumptions are credible due to controlled randomization or a natural experiment. This evidence can motivate the application of these assumptions in other contexts where similar identification strategies are not available. This contrasts with traditional IV assumptions, which tend to be highly context--specific.

The \cite{djebbari2008heterogeneous} study of the heterogeneous impacts of the PROGRESA conditional cash transfer program provides some examples of potential \ref{SPMIV}s. PROGRESA provided payments to households conditional on regular school attendance by the household's children as well as visits to health centers. \citeauthor{djebbari2008heterogeneous} find that the impact of this program on per capita consumption is substantially larger for poorer households and households in more ``marginal'' villages, i.e., villages with greater rates of illiteracy, more limited infrastructure, and a greater dependence on agricultural activities. Evaluations of cash transfer programs in other contexts could make use of this information by using household poverty or village marginality as \ref{SPMIV}s. 

Further examples are provided by the \cite{bitler2014can} study of the impact of the Connecticut Jobs First experiment. This program substantially lowered the marginal tax rate on earnings below the poverty line for families on relief, relative to the existing Aid to Families with Dependent Children (AFDC) program. In the Jobs First program, the entire benefit package is terminated once earnings rise above the poverty line; this is in contrast to the AFDC, where benefits decline linearly with earnings. Labor supply theory clearly suggests that the impact of this alternative budget scheme should boost earnings and employment much more for those who were previously out of work or whose earnings left them far below the poverty line. These hypotheses are strongly borne out by the data, suggesting that measures of pre--program earnings and employment could serve as \ref{SPMIV}s in studies of similar programs which are not implemented experimentally.

\cite{manski2000monotone} considered a specific case of an MIV where the average potential outcomes are monotone in the realized treatment; they referred to this as the monotone treatment selection assumption. A similar assumption can be developed in this context:

\begin{assumptionSPMTS}[Supermodular treatment selection]\label{SPMTS}
The \textit{supermodular treatment selection} assumption holds for $\E\left[\, y\left(t_{1}\right) - y\left(t_{2}\right) \,\right]$ with $t_{2}\leq t_{1}$ if
\begin{gather}\label{spmts 0}
\E\left[\, y\left(t_{1}\right) - y\left(t_{2}\right) \mid z = t_{1} \,\right] \geq \E\left[\, y\left(t_{1}\right) - y\left(t_{2}\right) \mid z = t_{2} \,\right]
\end{gather}
\end{assumptionSPMTS}

The bounds under this assumption will be analogous to those under \ref{SPMIV} derived below. Let $\underline{B}\left(t,x\right)$ and $\overline{B}\left(t,x\right)$ be defined as
\[
\underline{B}\left(t,x\right) = \E\left[\, y\left(t\right) \mid z=t, x \,\right] P\left(z=t \mid x\right) + \underline{K} P\left(z\neq t \mid x\right) \;\;\; \forall t\in T, x\in X
\]
and
\[
\overline{B}\left(t,x\right) = \E\left[\, y\left(t\right) \mid z=t, x \,\right] P\left(z=t \mid x \right) + \overline{K} P\left(z\neq t \mid x \right) \;\;\; \forall t\in T, x\in X
\]
The following bounds can be derived using \ref{SPMIV}:
\begin{proposition}\label{SPMIV 1}
Assume that $x^k$ is an \ref{SPMIV} for $\E\left[\, y\left(t_{1}\right) - y\left(t_{2}\right) \mid x^k,x^{-k} \,\right]$ with $t_{1}, t_{2}\in T$. Then, the bounds
\begin{equation}
\begin{gathered}\label{spmiv 1}
\sup_{x^{k}_{2} \leq x_{1}^k}\left\{\underline{B}\left(t_{1},x^{k}_{2}, x^{-k}\right) - \overline{B}\left(t_{2}, x^{k}_{2}, x^{-k}\right)\right\} \\
\leq \E\left[\, y\left(t_{1}\right) - y\left(t_{2}\right) \mid x_{1}^k,x^{-k} \,\right] \leq \\
\inf_{x_{1}^k \leq x^{k}_{2}}\left\{\overline{B}\left(t_{1}, x^{k}_{2}, x^{-k}\right) - \underline{B}\left(t_{2}, x^{k}_{2}, x^{-k}\right)\right\}
\end{gathered}
\end{equation}
are sharp.
\end{proposition}

As is the case for bounds derived under IV or MIV assumptions, inference is complicated by the $\sup$ and $\inf$ operators in equation \eqref{spmiv 1} \citep{manski2009more}. Analog estimators of the bounds in \eqref{spmiv 1} are consistent but biased in finite samples; the estimated bounds will generally be too narrow. Fortunately, the methods developed by \cite{chernozhukov2013intersection} can be applied to find bias--corrected estimates and associated confidence intervals. \citeauthor{chernozhukov2013intersection} discuss in detail the special cases of estimating nonparametric bounds using instrumental variables and MIVs; the bounds in \eqref{spmiv 1} are essentially identical for the purposes of estimation, so their results can be applied directly to my estimation problem. The theoretical extension allowing for multiple \ref{SPMIV}s is straightforward, and presents no novel estimation challenges besides those associated with high--dimensional nonparametric conditioning. 

Returning to assumption \ref{SPMIV}: If the second inequality in \eqref{spmiv 0} is reversed, $x^k$ becomes a \textit{submodular instrumental variable}. If $x^k$ is a supermodular and submodular instrumental variable, i.e., average treatment effects are constant across different values of $x^k$, then $x^k$ is a \textit{modular instrumental variable}. While this may seem like a strong assumption, it is routinely employed in applied work that models treatment effects without allowing for interactions. 

\ref{SPMIV}s may also improve the bounds on functionals of potential outcome distributions, as the following corollary illustrates:
\begin{corollary}\label{SPMIV 2}
Assume that $x^k$ is an \ref{SPMIV} for $\E\left[\, y\left(t_{1}\right) - y\left(t_{2}\right) \mid x^k,x^{-k} \,\right]$ with $t_{1}, t_{2}\in T$. Then, the bounds
\begin{equation}
\begin{gathered}\label{spmiv po 1}
\max\left\{\underline{B}\left(t_{1}, x_{1}^k, x^{-k}\right), \sup_{x^{k}_{2} \leq x_{1}^k}\left\{\underline{B}\left(t_{1},x^{k}_{2}, x^{-k}\right) - \overline{B}\left(t_{2}, x^{k}_{2}, x^{-k}\right)\right\} + \underline{B}\left(t_{2},x_{1}^k, x^{-k}\right)\right\} \\
\leq \E\left[\, y\left(t_{1}\right) \mid x_{1}^k,x^{-k} \,\right] \leq \\
\min\left\{\overline{B}\left(t_{1}, x_{1}^k, x^{-k}\right), \inf_{x_{1}^k \leq x^{k}_{2}}\left\{\overline{B}\left(t_{1}, x^{k}_{2}, x^{-k}\right) - \underline{B}\left(t_{2}, x^{k}_{2}, x^{-k}\right)\right\} + \overline{B}\left(t_{2},x_{1}^k, x^{-k}\right)\right\}
\end{gathered}
\end{equation}
and
\begin{equation}
\begin{gathered}\label{spmiv po 2}
\max\left\{\underline{B}\left(t_{2}, x_{1}^k, x^{-k}\right),  \underline{B}\left(t_{1},x_{1}^k, x^{-k}\right) - \inf_{x_{1}^k \leq x^{k}_{2}}\left\{\overline{B}\left(t_{1}, x^{k}_{2}, x^{-k}\right) - \underline{B}\left(t_{2}, x^{k}_{2}, x^{-k}\right)\right\}\right\} \\
\leq \E\left[\, y\left(t_{2}\right) \mid x_{1}^k,x^{-k} \,\right] \leq \\
\min\left\{\overline{B}\left(t_{2}, x_{1}^k, x^{-k}\right), \overline{B}\left(t_{1},x_{1}^k, x^{-k}\right) - \sup_{x^{k}_{2} \leq x_{1}^k}\left\{\underline{B}\left(t_{1},x^{k}_{2}, x^{-k}\right) - \overline{B}\left(t_{2}, x^{k}_{2}, x^{-k}\right)\right\} \right\}
\end{gathered}
\end{equation}
are sharp.
\end{corollary}

The proof of this result is straightforward, requiring only verification that these potential outcome bounds do not imply that the bounds in \eqref{spmiv 1} can be tightened. As in the case of proposition \ref{SPMIV 1}, the \citeauthor{chernozhukov2013intersection} approach to inference can be applied here. It is natural to consider the application of this IV assumption in concert with the shape restrictions considered above. Proposition \ref{SPMIV 3} below shows how \ref{SPMIV} can be combined with \ref{SPM} and SMTR:

\begin{proposition}\label{SPMIV 3}
Assume the conditions of proposition \ref{spm bounds 2} and assume that $x^k$ is an \ref{SPMIV} for $\E\left[\, y\left(t_{1}\right) - y\left(t_{2}\right) \mid x^k,x^{-k} \,\right]$. Let $\overline{\Delta  B}\left(t_{1},t_{2},x^k,x^{-k}\right)$ and $\underline{\Delta B}\left(t_{1},t_{2},x^k,x^{-k}\right)$ denote the conditional upper and lower bounds, respectively, for $\E\left[\, y\left(t_{1}\right) - y\left(t_{2}\right) \mid x^k,x^{-k} \,\right]$ from \eqref{sharp bounds 2}. Then, the bounds
\begin{equation}
\begin{gathered}\label{spmiv 3}
\sup_{x^{k}_{2} \leq x_{1}^k}\left\{\underline{\Delta B}\left(t_{1},t_{2},x_{2}^k,x^{-k}\right)\right\}
\leq \E\left[\, y\left(t_{1}\right) - y\left(t_{2}\right) \mid x_{1}^k,x^{-k} \,\right] \leq 
\inf_{x_{1}^k \leq x^{k}_{2}}\left\{\overline{\Delta B}\left(t_{1},t_{2},x^{k}_{2},x^{-k}\right)\right\}
\end{gathered}
\end{equation}
are sharp. Under the additional assumption of SMTR, these bounds are sharp when $\overline{\Delta  B}\left(t_{1},t_{2},x^k,x^{-k}\right)$ and $\underline{\Delta B}\left(t_{1},t_{2},x^k,x^{-k}\right)$ denote the conditional upper and lower bounds, respectively, from  \eqref{sharp bounds 3}.
\end{proposition}

This result, combining SPMIV with SMTR and the new shape restrictions proposed above, will be applied in the empirical exercise below.

\end{section}

\begin{section}{Independence}\label{Independence}

Independence assumptions have been used to operationalize the belief that individuals' realized treatments are unrelated to any individual characteristics which may influence responses. This should be the case, for example, in a randomized controlled trial. I show how statistical independence can be combined with shape restrictions and instrumental variables assumptions to narrow the bounds on entire treatment effect distributions. 

The familiar assumption of statistical independence of treatments and response functions is defined in my notation as follows:

\begin{assumptionSI}[Statistical independence]\label{SI}
Potential outcomes are \textit{statistically independent} of realized treatments if 
\[
P\left(y\left(t\right) \mid z\right) = P\left(y\left(t\right)\right) \;\;\; \forall t\in T
\]
\end{assumptionSI}

Assumption \ref{SI} implies that the marginal distribution of $y\left(t\right)$, denoted $F_{t}$, is point identified for all $t\in T$ such that $P\left(z = t\right) > 0$. However, the distribution of $y\left(t_{1}\right) - y\left(t_{2}\right)$, whose cumulative distribution function is denoted by $F_{t_{1},t_{2}}$, is only partially identified. \cite{makarov1982estimates} was the first to derive pointwise sharp bounds on the distribution of the sum of two random variables with fixed marginal distributions. \cite{frank1987best} derived these bounds in a simpler manner and extended them to allow for other operations such as differences and products as well as more than two variables. However, as \cite{kreinovich2006computing} show, these bounds are not sharp in the case of more than two variables. The following result, taken from Theorem 2 of \cite{williamson1990probabilistic}, gives the sharp bounds on the distribution of $y\left(t_{1}\right) - y\left(t_{2}\right)$ for any $t_{1},t_{2}\in T$:
\begin{equation*}
\begin{gathered}\label{makarov bounds} 
\underline{F}_{t_{1}, t_{2}}\left(w\right) = \underset{u+v=w}{\sup}\left\{\max\left\{F_{t_{1}}\left(u\right) - F_{t_{2}}\left(-v\right), 0\right\}\right\} \\
\leq F_{t_{1},t_{2}}\left(w\right) \leq \\
1 + \underset{u+v=w}{\inf}\left\{\min\left\{F_{t_{1}}\left(u\right) - F_{t_{2}}\left(-v\right), 0\right\}\right\} = \overline{F}_{t_{1}, t_{2}}\left(w\right)
\end{gathered}
\end{equation*}
\cite{fan2010sharp} discuss consistent nonparametric estimation of these bounds. This result can be applied to derive sharp bounds on the quantile function of $y\left(t_{1}\right) - y\left(t_{2}\right)$, $F_{t_{1}, t_{2}}^{-1}\left(q\right)$, where $F_{t_{1}, t_{2}}^{-1}$ denotes the generalized inverse of the cdf $F_{t_{1}, t_{2}}$. Define the following functions:

\begin{equation}
\begin{gathered}\label{quantile function} 
\underline{F}_{t_{1}, t_{2}}^{-1}\left(q\right) = \begin{cases} \underset{u\in\left[q,1\right]}{\inf}\left[F_{t_1}^{-1}\left(u\right) - F_{t_2}^{-1}\left(u - q\right)\right] & \text{if } q \neq 0 \\ 
F_{t_1}^{-1}\left(0\right) - F_{t_2}^{-1}\left(1\right) & \text{if } q = 0 \end{cases} \\
\overline{F}_{t_{1}, t_{2}}^{-1}\left(q\right) = \begin{cases} \underset{u\in\left[0,q\right]}{\sup}\left[F_{t_1}^{-1}\left(u\right) - F_{t_2}^{-1}\left(1 + u - q\right)\right] & \text{if } q \neq 1 \\ 
F_{t_1}^{-1}\left(1\right) - F_{t_2}^{-1}\left(0\right) & \text{if } q = 1 \end{cases}
\end{gathered}
\end{equation}

Then, $\overline{F}_{t_{1}, t_{2}}^{-1}\left(q\right) \leq F_{t_{1}, t_{2}}^{-1}\left(q\right) \leq \underline{F}_{t_{1}, t_{2}}^{-1}\left(q\right)$. \ref{SI} can be combined with \ref{SPM} to refine these bounds, as the following result shows:
\begin{proposition}\label{te si}
Assume that \ref{SI} holds and that $T = \left\{t_{1}\wedge t_{2}, t_{1}, t_{2}, t_{1}\vee t_{2}\right\}$ with $t_{1}\wedge t_{2} < t_{1}, t_{2} < t_{1}\vee t_{2}$. Assume that \ref{SPM} holds on $T$. Then, the bounds
\begin{gather*}\label{si bounds 1}
\max\left\{\overline{F}_{t_{1}\vee t_{2}, t_{1}}^{-1}\left(q\right), \overline{F}_{t_{2}, t_{1}\wedge t_{2}}^{-1}\left(q\right)\right\} \leq F_{t_{1}\vee t_{2}, t_{1}}^{-1}\left(q\right) \leq \underline{F}_{t_{1}\vee t_{2}, t_{1}}^{-1}\left(q\right)
\end{gather*}
and
\begin{gather*}\label{si bounds 2}
\overline{F}_{t_{1}, t_{1}\wedge t_{2}}^{-1}\left(q\right)
\leq F_{t_{1}, t_{1}\wedge t_{2}}\left(w\right) \leq
\min\left\{\underline{F}_{t_{1}, t_{1}\wedge t_{2}}^{-1}\left(q\right), \underline{F}_{t_{1}\vee t_{2}, t_{2}}^{-1}\left(q\right)\right\}
\end{gather*}
are sharp.
\end{proposition}

Similar results can be derived for \ref{SBM}. These shape restrictions could be justified by theoretical arguments; alternatively, since average treatment effects are point--identified in this context, the supermodularity or submodularity of average effects could be used to provide some justification for stronger structural assumptions. Extending these results to general lattices is problematic due to the fact that sharp bounds on the distribution function of a sum of more than two variables are an open question. Nonetheless, it is straightforward to collect all possible stochastic dominance relations implied by the maintained assumptions, and bounds which contain the true value (but are not necessarily sharp) can be obtained in a manner similar to that in proposition \ref{te si}. Such bounds may be useful in policy evaluation.

A reformulation of the \ref{SPMIV} assumption can also be applied in this setting:\footnote{\ref{SPMIV} itself is unhelpful, since conditional average treatment effects are point--identified.}

\begin{assumptionQSPMIV}[Quantile supermodular instrumental variable]\label{QSPMIV}
$x^k$ is a \textit{quantile supermodular instrumental variable} for $y\left(t_{1}\right) - y\left(t_{2}\right)$ if
\begin{gather*}
x_{1}^k \leq x^{k}_{2}  
\Longrightarrow F_{t_{1},t_{2}}^{-1}\left(q \mid x_{1}^k, x^{-k}\right) \leq F_{t_{1},t_{2}}^{-1}\left(q \mid x^{k}_{2}, x^{-k}\right)
\end{gather*}
for all $x^{-k}$.
\end{assumptionQSPMIV}

\cite{giustinelli2011non} analyzes the returns to education in Italy using a similar restriction on the quantile function of potential outcomes. \cite{blundell2007changes} impose monotonicity in a covariate on the conditional cdf of a potential outcome. Their bounds are simplified by the fact that the distribution of potential outcomes is partially observed, while the distribution of treatment effects is never observed, necessitating the use of the \cite{williamson1990probabilistic} bounds. The following proposition computes the bounds derived under \ref{QSPMIV}:
\begin{proposition}\label{qspmiv}
Assume that \ref{SI} holds. Assume that $x^k$ is a \ref{QSPMIV} for $y\left(t_{1}\right) - y\left(t_{2}\right)$ with $t_{1}, t_{2}\in T$. Then, the bounds
\begin{equation*}
\begin{gathered}\label{qspmiv 1}
\sup_{x^{k}_{2}\leq x_{1}^k}\left\{\overline{F}_{t_{1}, t_{2}}^{-1}\left(q \mid x^{k}_{2}, x^{-k}\right) \right\} \leq F_{t_{1},t_{2}}^{-1}\left(q \mid x_{1}^k, x^{-k}\right) \leq \inf_{x_{1}^k\leq x^{k}_{2}}\left\{\underline{F}_{t_{1}, t_{2}}^{-1}\left(q \mid x^{k}_{2}, x^{-k}\right)\right\}
\end{gathered}
\end{equation*}
are sharp.
\end{proposition}

Again, since conditional average treatment effects are point--identified, they can provide some evidence to support the validity of the stronger \ref{QSPMIV} assumption.

\end{section}

\begin{section}{Empirical Illustration}\label{emp}

To illustrate the use of the identification results developed in this paper, I reanalyze data from  \cite{shertzer2016zoning}. That study examines the extent to which Chicago's first zoning ordinance, passed in 1923, influenced the evolution of the spatial distribution of commercial, industrial, and residential activity in the city. Evidence of substantial treatment effect heterogeneity was found, motivating the use of \ref{SBM} and \ref{SPMIV} assumptions in the analysis below.

Chicago's 1923 zoning ordinance regulated land by restricting uses and density.\footnote{For details on the ordinance, consult \cite{shertzer2016zoning}.} Here, I bound the effects of 1923 commercial zoning on the probability that a city block will contain any commercial activity in 2005. Zoning for each type of use (commercial, industrial, single/multi--family residential) was coupled with an allowed maximum density level. Focusing on four different allowed density levels and an indicator for whether or not a block received commercial zoning yields eight distinct treatments and four distinct commercial zoning treatment effects.

Formally, the outcome variable $y\left(\cdot\right)$ is an indicator equal to 1 iff city block $i$ contains any commercial activity in 2005. $y$ is a function of a treatment $t \in T =  \left\{0,1\right\}\times\left\{1,2,3,4\right\}$. The first dimension of $t$ is equal to 1 if the block received any commercial zoning in 1923 and 0 otherwise. The second dimension of $t$ is equal to 1 if the block was zoned for the lowest density development (3 or fewer stories), 2 if it was zoned for higher densities suitable for apartment buildings (8--10 stories), 3 if it was zoned for high--rise commercial buildings, and 4 if it was zoned for the tallest commercial skyscrapers (this latter zoning was reserved for the central business district and immediately surrounding area).  

Since numerous other factors, such as pre--zoning land use, property values, and demographics, shaped both the initial zoning ordinance as well as the future development of the city, the exogenous treatment selection (ETS) assumption is likely too strong. As an alternative, I employ mixtures of SBM, SMTR, and SPMIV in the empirical analysis. \ref{SBM} implies that the commercial zoning treatment effect will be larger when paired with lower density zoning. This is motivated by the fact that areas zoned for lower densities will be more residential in character and contain a larger proportion of single--family homes \citep{shertzer2016zoning}. It is well documented that residential property owners (especially single--family homeowners) generally oppose the encroachment of commercial uses and have substantial power to block such development \citep{fischel2001homevoter}. It is likely that the early establishment of commercial activity through zoning will be a more important determinant of future commercial land use in areas also zoned for lower densities.\footnote{\cite{shertzer2016zoning} provide direct evidence of the veracity of this assumption.} This assumption is also consistent with previous literature showing that mixed use areas are more likely to see conversion to completely non--residential use than strictly residential use \citep{mcmillen1991markov}. 

While the effect of commercial zoning is likely to be larger in low--density areas, it is generally the case that commercial uses appear more often in denser areas. Commercial uses can afford the higher rents that prevail in dense areas, benefit more from agglomeration economies, and are compatible with residential development through mixed--use structures. This, combined with the fact that commercial zoning is unlikely to lower the likelihood of commercial development, suggests that SMTR is an appropriate assumption in this context. Table \ref{tab1} shows how the upper and lower bounds on the average treatment effect of commercial zoning, $\E\left[\, y\left(1,d\right) - y\left(0,d\right) \,\right]$, vary under different shape restrictions; bounds on the ATE are obtained for each density level $d\in\left\{1,2,3,4\right\}$. Appendix figure \ref{fig1} depicts the results from table \ref{tab1} graphically.

\begin{table}[ht]
\begin{center}
\caption{Bounds on the Effects of Commercial Zoning on Future Commercial Use}
\small
\begin{tabular}{lccccccccc} 
\toprule
 & \multicolumn{9}{c}{\centering Assumptions} \\
\cmidrule{2-10}
\multicolumn{1}{l}{\centering Treatment effect} & \multicolumn{1}{c}{\centering ETS} & \multicolumn{2}{c}{\centering None} & \multicolumn{2}{c}{\centering SBM} & \multicolumn{2}{c}{\centering SMTR} & \multicolumn{2}{c}{\centering \ref{SBM}/SMTR} \\
\cmidrule{2-4} \cmidrule{5-6} \cmidrule{7-8} \cmidrule{9-10} 
& &  \multicolumn{1}{c}{\centering Lower} & \multicolumn{1}{c}{\centering Upper} & \multicolumn{1}{c}{\centering Lower} & \multicolumn{1}{c}{\centering Upper} & \multicolumn{1}{c}{\centering Lower} & \multicolumn{1}{c}{\centering Upper} & \multicolumn{1}{c}{\centering Lower} & \multicolumn{1}{c}{\centering Upper} \\
\midrule
\multirow{2}{*}{$\E\left[\, y\left(1,1\right) - y\left(0,1\right) \,\right]$}  & 0.579  &    -0.702 &     0.932 &    -0.271 &     0.932 &     0 &     0.754 &     0 &     0.754 \\  
& (0.557, 0.601) &    -0.710 &     0.936 &    -0.277 &     0.936 &     0 &     0.761 &     0 &     0.761 \\
\multirow{2}{*}{$\E\left[\, y\left(1,2\right) - y\left(0,2\right) \,\right]$} &  0.503  &    -0.664 &     0.858 &    -0.569 &     0.790 &     0 &     0.803 &     0 &     0.738 \\  
& (0.483, 0.524) &    -0.672 &     0.863 &    -0.578 &     0.796 &     0 &     0.809 &     0 &     0.745 \\ 
\multirow{2}{*}{$\E\left[\, y\left(1,3\right) - y\left(0,3\right) \,\right]$} & 0.311  &    -0.936 &     0.957 &    -0.905 &     0.747 &     0 &     0.923 &     0 &     0.732 \\  
& (0.269, 0.354) &    -0.941 &     0.961 &    -0.910 &     0.754 &     0 &     0.927 &     0 &     0.739 \\  
\multirow{2}{*}{$\E\left[\, y\left(1,4\right) - y\left(0,4\right) \,\right]$} &  0.341 &    -0.969 &     0.982 &    -0.969 &     0.729 &     0 &     0.957 &     0 &     0.729 \\  
& (0.266, 0.417) &    -0.972 &     0.984 &    -0.972 &     0.736 &     0 &     0.960 &     0 &     0.736 \\ 
 \midrule
Observations & 14,690 & 14,690 & 14,690 & 14,690 & 14,690 & 14,690 & 14,690 & 14,690 & 14,690 \\
\bottomrule
\end{tabular}
\caption*{\small Bounds on the effect of 1923 commercial zoning on the probability of commercial use in 2005 under different shape restrictions. Commercial zoning effect is estimated separately for all four levels of 1923 density zoning. Upper and lower bounds are accompanied below by bootstrapped upper and lower 95\% confidence intervals. Column 1 presents point estimates assuming exogenous treatment selection (ETS), along with 95\% confidence intervals. Columns 2--3 present the no assumption bounds. Columns 4--5 add the assumption of submodularity (SBM). Columns 6--7 add the assumption of semi--monotone treatment response (SMTR). Columns 8--9 combine \ref{SBM} and SMTR.}
\label{tab1}
\end{center}
\end{table}

The results in table \ref{tab1} show how the imposition of the shape restrictions \ref{SBM} and SMTR can separately and jointly improve upon the no assumption upper and lower bounds on treatment effects. The no assumption bounds are very wide due to the large number of treatments. The imposition of \ref{SBM} establishes monotonicity of the upper and lower bounds as expected; upper bounds on the commercial zoning effects are smaller in higher density areas, while lower bounds are higher in low density areas. SMTR restricts the lower bound to be zero but its effect on upper bounds is theoretically ambiguous. Upper bounds on the potential outcomes treated with commercial zoning are monotone increasing with density:
\[
\E\left[\, y\left(1,1\right) \,\right]\leq \E\left[\, y\left(1,2\right) \,\right] \leq \E\left[\, y\left(1,3\right) \,\right] \leq \E\left[\, y\left(1,4\right) \,\right]
\]
However, the same is true of lower bounds on potential outcomes of observations not treated with commercial zoning:
\[
\E\left[\, y\left(0,1\right) \,\right]\leq \E\left[\, y\left(0,2\right) \,\right] \leq \E\left[\, y\left(0,3\right) \,\right] \leq \E\left[\, y\left(0,4\right) \,\right]
\]
In this case, the SMTR upper bounds on treatment effects are increasing with density; however, with the addition of \ref{SBM}, this is reversed, since \ref{SBM} implies that the commercial zoning effect will be larger in areas also zoned for low densities. Under both SMTR and \ref{SBM}, the results indicate that commercial zoning in 1923 increases the probability of commercial activity in 2005 by at most roughly 0.75 on the entire sample. While this still leaves a sizable range of possible values, it shows that historical zoning is not necessarily decisive, a fact which is not discernible from the uninformative no assumption bounds or the SMTR bounds in high density areas. 

\begin{table}[ht]
\caption{Bounds on the Effect of Commercial Zoning on Future Commercial Use: Subsample Not Treated with Commercial Zoning}
\begin{center}
\small
\begin{tabular}{lcccccc} 
\toprule
 & \multicolumn{6}{c}{\centering Assumptions} \\
\cmidrule{2-7}
\multicolumn{1}{l}{\centering Treatment effect} & \multicolumn{2}{c}{\centering None} & \multicolumn{2}{c}{\centering SBM} & \multicolumn{2}{c}{\centering \ref{SBM}/SPMIV} \\
\cmidrule{2-3} \cmidrule{4-5} \cmidrule{6-7}
&  \multicolumn{1}{c}{\centering Lower} & \multicolumn{1}{c}{\centering Upper} & \multicolumn{1}{c}{\centering Lower} & \multicolumn{1}{c}{\centering Upper} & \multicolumn{1}{c}{\centering Lower} & \multicolumn{1}{c}{\centering Upper} \\
\midrule
\multirow{2}{*}{$\E\left[\, y\left(1,1\right) - y\left(0,1\right) \,\right]$}  & -0.509  & 0.991 & -0.068  & 0.991 & -0.01 &  0.992 \\  
&  -0.521 & 0.994 & -0.075  & 0.994 & -0.0213  &  0.994 \\  
 \midrule
Observations & 6,072 & 6,072 & 6,072 & 6,072 & 6,027 & 6,027 \\
\midrule
Sample & \multicolumn{2}{c}{\centering No comm. zoning} & \multicolumn{2}{c}{\centering No comm. zoning} & \multicolumn{2}{c}{\centering No comm. zoning} \\
\bottomrule
\end{tabular}
\caption*{\small Bounds on the effect of 1923 commercial zoning on the probability of commercial use in 2005 when paired with the most restrictive 1923 density zoning. Sample is restricted to blocks not treated with commercial zoning. No assumption and \ref{SBM} upper and lower bounds are reported with bootstrapped upper and lower 95\% confidence intervals. The average treatment effect in the final two columns is conditional on 40\% of the population being white (3rd gen. and above) with the SPMIV assumption that conditional average treatment effects are monotone increasing in the percentage of the population that are non--3rd gen. white. \ref{SPMIV} bounds are half--median unbiased estimates with associated 95\% confidence intervals calculated using the \cite{chernozhukov2015implementing} Stata implementation of the inference methods of \cite{chernozhukov2013intersection}.}
\label{tab2}
\end{center}
\end{table}

While \ref{SBM} cannot establish the nonnegativity of the treatment effects on the whole sample, an analysis of the impact of commercial and low density zoning on the subsample of blocks which were not treated with commercial zoning is more informative. Table \ref{tab2} shows that \ref{SBM} alone establishes a lower bound of -0.068 on $\E\left[\, y\left(1,1\right) - y\left(0,1\right) \,\right]$. This can be strengthened using the fraction of the block's population that are not at least 3rd generation native white as an \ref{SPMIV}. The argument here is that native whites were a politically powerful group relative to blacks and recent immigrants, especially those from Eastern Europe. These latter groups would have less power to alter zoning ex post, and so one would expect the 1923 zoning effect to be larger if more of the population belongs to these politically marginalized groups.\footnote{This is consistent with the historical record as well as the results of \cite{shertzer2016race}, which provide direct evidence of discrimination against blacks and recent immigrants in the drafting of the 1923 ordinance in Chicago. An alternative approach could involve including recent Irish immigrants with native whites, as the Irish became politically influential early in Chicago's history; experimenting with this alternate approach yielded similar results.} Computing the average treatment effect conditional on 40\% of the population being white using this \ref{SPMIV} yields a lower bound of -0.01, so these two assumptions alone nearly identify the sign of the treatment effect.

\begin{table}[ht]
\caption{Bounds on the Effect of Commercial Zoning on Future Commercial Use: Undeveloped Subsample}
\begin{center}
\small
\begin{tabular}{lccccc} 
\toprule
&  \multicolumn{5}{c}{\centering Assumptions} \\
\cmidrule{2-6}
\multicolumn{1}{l}{\centering Treatment effect} & \multicolumn{1}{c}{\centering ETS} & \multicolumn{1}{c}{\centering None} & \multicolumn{1}{c}{\centering SMTR} & \multicolumn{1}{c}{\centering \ref{SBM}/SMTR} & \multicolumn{1}{c}{\centering SBM/SMTR/\ref{SPMIV}} \\
\cmidrule{2-6}
&  & \multicolumn{1}{c}{\centering Upper} & \multicolumn{1}{c}{\centering Upper}  & \multicolumn{1}{c}{\centering Upper} & \multicolumn{1}{c}{\centering Upper} \\
\midrule
\multirow{2}{*}{$\E\left[\, y\left(1,1\right) - y\left(0,1\right) \,\right]$} & 0.51 & 0.846 & 0.806 & 0.806 &  0.661 \\ 
& (0.469, 0.550)  & 0.862 & 0.825 & 0.825 & 0.688 \\  
\multirow{2}{*}{$\E\left[\, y\left(1,2\right) - y\left(0,2\right) \,\right]$} & 0.425 & 0.921 & 0.918 &  0.769 & 0.628 \\ 
& (0.352, 0.496) & 0.931 & 0.928 & 0.786 &  0.656 \\  
 \midrule
Observations & 3,669 & 3,669 & 3,669 & 3,669 & 3,669 \\
\midrule
Sample & \multicolumn{1}{c}{\centering Undeveloped} & \multicolumn{1}{c}{\centering Undeveloped} & \multicolumn{1}{c}{\centering Undeveloped} & \multicolumn{1}{c}{\centering Undeveloped} & \multicolumn{1}{c}{\centering Undeveloped} \\
\bottomrule
\end{tabular}
\caption*{\small Bounds on the effect of 1923 commercial zoning on the probability of commercial use in 2005 when paired with the two most restrictive levels of 1923 density zoning. Estimation is on a subset of blocks that were largely undeveloped in 1923; since these areas exclusively received zoning for either density level 1 or 2, I only discuss those treatment effects. Average treatment effects are conditional on 40\% of the population being white (3rd gen. and above). The bounds are computed under no assumptions, \ref{SBM}/SMTR, and \ref{SBM}/SMTR/\ref{SPMIV}. No assumption and \ref{SBM}/SMTR upper bounds are reported with upper 95\% confidence intervals from a kernel--weighted local polynomial smoother. SPMIV bounds assume that conditional average treatment effects are monotone increasing in the percentage of the population that are non--3rd gen. white. \ref{SPMIV} bounds are half--median unbiased estimates with associated 95\% confidence intervals calculated using the \cite{chernozhukov2015implementing} Stata implementation of the inference methods of \cite{chernozhukov2013intersection}.}
\label{tab3}
\end{center}
\end{table}

Table \ref{tab3} presents a further analysis of the upper bounds of the commercial zoning effect. The sample is restricted to blocks that fall into the bottom quartile of pre--zoning development; development is measured using a continuous index which is a function of population density, building heights, commercial and manufacturing establishment density, and distance to the central business district and Lake Michigan. The blocks in this sample lie in the outlying areas of the city which were largely undeveloped in 1923, with many vacant lots, very low population density, and few business establishments.\footnote{These blocks had a median of two residents and zero commercial/industrial uses per acre.} This is the area where the effect of commercial zoning should be the largest, as these areas also received the most restrictive density zoning and the pattern of future development was essentially undetermined. 

The results in table \ref{tab3} show that commercial zoning in 1923 increases the probability of commercial activity in 2005 by at most 0.66 in areas zoned for the lowest density and 0.63 in areas zoned for higher densities. While still compatible with a range of magnitudes, these results show that historical zoning was limited in its ability to direct future development, leaving a large role for market--driven rezonings to alter the initial plan.

\end{section}

\begin{section}{Conclusion}

In this paper, I contribute to the literature on the partial identification of treatment effects by developing and applying assumptions that formalize the notion of complementarity. I examine the identification power of these assumptions and discuss how they can be justified. The supermodularity and submodularity assumptions proposed can be used to narrow bounds on treatment effects in studies of policy complementarity, which have traditionally been stymied by a lack of pseudo--experimental variation in multiple policies simultaneously. Proposition \ref{spm bounds 0} shows how these shape restrictions can improve bounds on average treatment effects in the simple case of two binary treatments. Proposition \ref{spm bounds 2} extends this result to a more general treatment set with an arbitrary finite number of (possibly multivalued) treatments and the possibility of complex combinations of supermodularity and submodularity. Proposition \ref{spm bounds 3} combines the supermodularity and submodularity assumptions with the semi--monotone treatment response assumption of \cite{manski1997monotone}.

Complementarity may also stem from differential treatment response among subpopulations defined by observed covariates. Subgroup heterogeneity in treatment effects is an increasingly widely recognized phenomenon, and can often be motivated directly from economic theory (see, e.g., \cite{bitler2014can}). Proposition \ref{SPMIV 1} and corollary \ref{SPMIV 2} show how qualitative information about treatment effect heterogeneity embodied in supermodular instrumental variables can be used to improve bounds on average treatment effects and average potential outcomes, respectively. Proposition \ref{SPMIV 3} combines supermodular IVs with the new shape restrictions as well as semi--monotone treatment response. Supermodular instrumental variables can be used in studies with one or many treatments, making them a versatile and potentially powerful addition to the arsenal of applied econometricians.

The assumptions developed here can be useful in the experimental context as well. Proposition \ref{te si} shows how supermodularity can be combined with an assumption of statistical independence between assigned treatments and responses to yield improved bounds on the cumulative distribution function of a treatment effect. These results can be applied to the evaluation of outcomes in complex (multi--treatment) randomized controlled trials, which are increasingly prevalent in many fields, including development economics. Since average treatment effects are point--identified in this context, one can determine if average responses exhibit supermodularity or submodularity. This can provide evidence that individual response functions are supermodular or submodular. Similarly, the behavior of (point--identified) conditional average treatment effects can motivate the use of a quantile supermodular instrumental variable; proposition \ref{qspmiv} shows how this assumption can strengthen the bounds on the quantile function of a treatment effect distribution. 

Bounds derived under the assumptions I propose here are of interest only to the extent that such assumptions are considered credible. Where might evidence for their validity come from? Arguments for policy complementarity may be provided by economic theory, as in \cite{lalive2006how}, or they may come from multi--treatment randomized controlled trials. Evidence on subgroup heterogeneity in treatment effects may be provided by previous studies where strong identifying assumptions are credible due to controlled randomization or a natural experiment. In such studies, conditional average treatment effects are point--identified, so the validity of the proposed assumptions can be established. This can motivate their use in other contexts where similar identification strategies are not available. This distinguishes supermodular IV assumptions from traditional IV assumptions, since the latter tend to be context--specific.

The empirical illustration in section \ref{emp} employs submodularity, semi--monotone treatment response, and supermodular IVs to study the impact of historical zoning on the evolution of land use in Chicago. In some cases, I am able to essentially rule out the possibility of a negative commercial zoning effect using only submodularity and a supermodular IV. I am also able to show that the effect of commercial zoning was far from decisive even in the areas of the city that were undeveloped at the time of zoning, demonstrating the potentially sizable role of market--driven zoning revisions.

\end{section}

\singlespacing

\bibliographystyle{econ}
\bibliography{biblio}

\newpage

\doublespacing

\begin{section}{Appendix: Figures}

\begin{figure}[ht]
  \centering
     \caption{Bounds on the Effects of Commercial Zoning on Future Commercial Use}
\subfloat[]{%
  \centering%
    \includegraphics[width=0.38\textwidth]{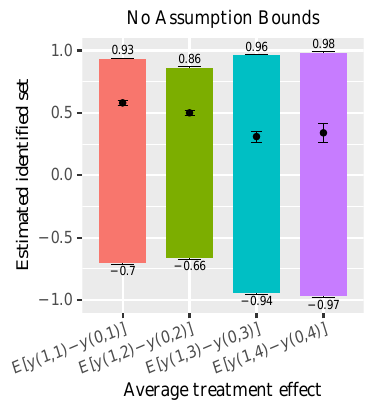}}
    \quad
    \subfloat[]{%
  \centering%
    \includegraphics[width=0.38\textwidth]{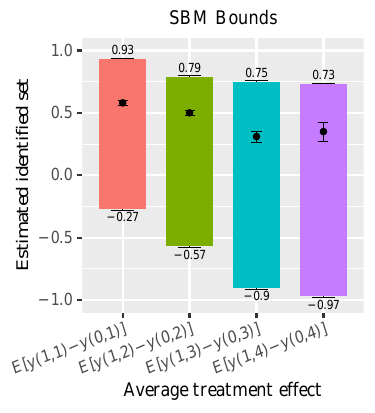}
    }
    
\subfloat[]{%
  \centering%
    \includegraphics[width=0.38\textwidth]{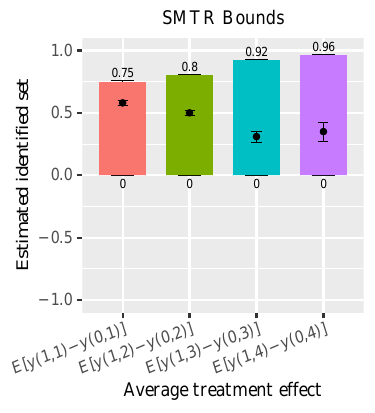}}
    \subfloat[]{%
  \centering%
  \includegraphics[width=0.38\textwidth]{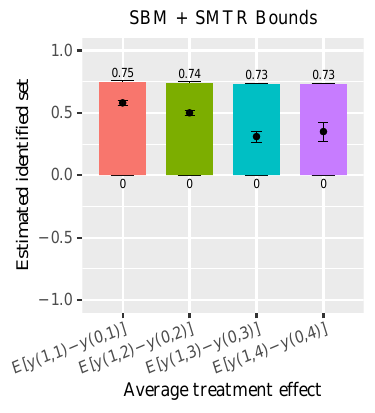} 
    }
    \label{fig1}
   \caption*{\small Bounds on the effect of 1923 commercial zoning on the probability of commercial use in 2005 under a variety of shape restrictions. The commercial zoning effect is estimated separately for all four levels of 1923 density zoning; colors denote the level of density zoning. Upper and lower bounds are accompanied by upper and lower 95\% confidence bounds, respectively. Confidence intervals estimated using a standard bootstrap. Point estimates from a naive OLS regression are included along with 95\% confidence intervals. Panel A shows the no assumption bounds. Panel B adds the assumption of submodularity (SBM). Panel C adds the assumption of semi--monotone treatment response (SMTR). Panel D combines \ref{SBM} and SMTR.}
\end{figure}

\end{section}

\newpage

\begin{section}{Appendix: Proofs}

\begin{proof}[Proof of lemma \ref{lattice lemma}]
I first show that $t_{1}\wedge t_{5} = t_{2}$:
\begin{align*}
t_{5}\wedge t_{4} & = t_{3} \\
\Longrightarrow t_{1}\wedge\left(t_{4}\wedge t_{5}\right) & = t_{1} \wedge t_{3} \\
\Longrightarrow \left(t_{1}\wedge t_{4}\right)\wedge t_{5} & = t_{2} \\
\Longrightarrow t_{1}\wedge t_{5} & = t_{2}
\end{align*}
where the last implication follows from the fact that $t_{1}\leq t_{4}$. Now, I show that $t_{1}\vee t_{5} = t_{6}$:
\begin{align*}
t_{4}\vee t_{5} & = t_{6} \\
\Longrightarrow \left(t_{1} \vee t_{3}\right)\vee t_{5} & = t_{6} \\
\Longrightarrow t_{1} \vee \left(t_{3}\vee t_{5}\right) & = t_{6} \\
\Longrightarrow t_{1} \vee t_{5} & = t_{6} 
\end{align*}
where the last implication follows from the fact that $t_{3}\leq t_{5}$.
\end{proof}

\begin{proof}[Proof of proposition \ref{spm bounds 3}]

This proof requires only a minor modification of the proof of proposition \ref{spm bounds 2}. Following \cite{manski1997monotone}, SMTR implies $y\left(t_{4}\right)\leq y\left(t_{3}\right)$ and thus $y\left(t_{3}\right) - y\left(t_{4}\right) > 0$ for all $t_{4} < t_{3}$, and this lower bound is sharp under SMTR alone. The lower bounds in \eqref{spm bounds 2 proof 1} take the form $y\left(z\right) - \overline{K}$, $\underline{K} - y\left(z\right)$, and $\underline{K} - \overline{K}$; these are necessarily weakly negative regardless of the value of $y\left(z\right)$, so they do not improve upon the SMTR lower bound. Thus, a sharp lower bound for the individual treatment effect of interest is
\[
y\left(t_{1}\right) - y\left(t_{2}\right) \geq 0
\]
Turning to the upper bounds, if $z \geq t_{1}$, then SMTR implies
\[
y\left(t_{1}\right) - y\left(t_{2}\right) \leq y\left(z\right) - \underline{K}
\]
This is the same upper bound implied by \ref{SPM} when $z\in \Lambda_{1} \cup \Lambda_{2}$, so SMTR has no additional identification power in this case. If $z\geq t_{1}$ and $z\in \Lambda_{3} \cup \Lambda_{13}$, SMTR generally improves upon the no--assumption upper bound $\overline{K} - \underline{K}$. If $z \leq t_{2}$, then SMTR implies
\[
y\left(t_{1}\right) - y\left(t_{2}\right) \leq \overline{K} - y\left(z\right)
\]
If $z\in \Lambda_{4} \cup \Lambda_{6}$, \ref{SBM} implies the same upper bound, so SMTR has no additional identifying power. If $z\leq t_{2}$ and $z\in \Lambda_{5} \cup \Lambda_{13}$, SMTR improves upon the no--assumption upper bound. If $z \in \Lambda_{10} \cup \Lambda_{11}$, then SMTR has no implications for the upper bound on $y\left(t_{1}\right) - y\left(t_{2}\right)$, but \ref{SPM} does via \ref{spm bounds 2 proof 1}. If $z \in \Lambda_{7} \cup \Lambda_{9}$, SMTR also has no implications for the upper bound on $y\left(t_{1}\right) - y\left(t_{2}\right)$, but \ref{SBM} weakly improves upon it, again via \ref{spm bounds 2 proof 1}. Combining these results yields
\begin{equation}\label{spm bounds 3 proof 1}
y\left(t_{1}\right) - y\left(t_{2}\right) \in 
\begin{cases}
\left[0, \overline{K} - y\left(z\right)\right] & \text{if } z \in \left\{t_{3} \mid t_{3}\leq t_{2}\right\} \cup \Lambda_{10} \cup \Lambda_{11} \\
\left[0, y\left(z\right) - \underline{K}\right] & \text{if } z \in \left\{t_{3} \mid t_{1}\leq t_{3}\right\} \cup \Lambda_{7} \cup \Lambda_{9} \\
\left[0, \overline{K} - \underline{K}\right] & \text{if } z \in \left\{t_{3} \mid t_{1}\nleq t_{3} \text{ and } t_{3}\nleq t_{2}\right\} \cap \left(\Lambda_{8} \cup \Lambda_{12} \cup \Lambda_{13}\right) \\
\end{cases}
\end{equation}

Averaging over $i$ yields the result in \eqref{sharp bounds 3}.

\end{proof}

\begin{proof}[Proof of proposition \ref{SPMIV 1}]
In the absence of other assumptions, the bounds
\[
\underline{B}\left(t_{1},x\right) \leq \E\left[\, y\left(t_{1}\right) \mid x \,\right] \leq \overline{B}\left(t_{1},x\right)
\]
and
\[
\underline{B}\left(t_{2},x\right) \leq \E\left[\, y\left(t_{2}\right) \mid x \,\right] \leq \overline{B}\left(t_{2},x\right)
\]
and thus
\[
\underline{B}\left(t_{1},x\right) - \overline{B}\left(t_{2},x\right) \leq \E\left[\, y\left(t_{1}\right) \mid x \,\right] - \E\left[\, y\left(t_{2}\right) \mid x \,\right] \leq \overline{B}\left(t_{1},x\right) - \underline{B}\left(t_{1},x\right)
\]
are sharp for all $x\in X$. The assumption that $x^k$ is an \ref{SPMIV} for $\E\left[\, y\left(t_{1}\right) - y\left(t_{2}\right) \mid x^k,x^{-k} \,\right]$ implies that
\[
\underline{B}\left(t_{1},x^{k}_{2}, x^{-k}\right) - \overline{B}\left(t_{2},x^{k}_{2}, x^{-k}\right) \leq \E\left[\, y\left(t_{1}\right) - y\left(t_{2}\right) \mid x_{1}^k,x^{-k} \,\right]
\]
for all $x^{k}_{2} \leq x_{1}^k$ and
\[
\E\left[\, y\left(t_{1}\right) - y\left(t_{2}\right) \mid x_{1}^k,x^{-k} \,\right] \leq \underline{B}\left(t_{1},x^{k}_{2}, x^{-k}\right) - \overline{B}\left(t_{2},x^{k}_{2}, x^{-k}\right) 
\]
for all $x_{1}^k \leq x^{k}_{2}$. The result follows.
\end{proof}

\begin{proof}[Proof of proposition \ref{SPMIV 2}]
Proposition \ref{SPMIV 1} implies that the bounds
\begin{gather*}
\sup_{x^{k}_{2} \leq x_{1}^k}\left\{\underline{B}\left(t_{1},x^{k}_{2}, x^{-k}\right) - \overline{B}\left(t_{2}, x^{k}_{2}, x^{-k}\right)\right\} \\
\leq \E\left[\, y\left(t_{1}\right) - y\left(t_{2}\right) \mid x_{1}^k,x^{-k} \,\right] \leq \\
\inf_{x_{1}^k \leq x^{k}_{2}}\left\{\overline{B}\left(t_{1}, x^{k}_{2}, x^{-k}\right) - \underline{B}\left(t_{2}, x^{k}_{2}, x^{-k}\right)\right\}
\end{gather*}
are sharp. Thus, $\E\left[\, y\left(t_{1}\right) \mid x \,\right]$ and $\E\left[\, y\left(t_{2}\right) \mid x \,\right]$ must simultaneously satisfy the no--assumption bounds
\begin{equation}\label{spmiv na 1}
\underline{B}\left(t_{1}, x\right) \leq \E\left[\, y\left(t_{1}\right) \mid x \,\right] \leq \overline{B}\left(t_{1}, x\right)
\end{equation}
and
\begin{equation}\label{spmiv na 2}
\underline{B}\left(t_{2}, x\right) \leq \E\left[\, y\left(t_{2}\right) \mid x \,\right] \leq \overline{B}\left(t_{2}, x\right)
\end{equation}
as well as
\begin{equation}
\begin{gathered}
\sup_{x^{k}_{2} \leq x_{1}^k}\left\{\underline{B}\left(t_{1},x^{k}_{2}, x^{-k}\right) - \overline{B}\left(t_{2}, x^{k}_{2}, x^{-k}\right)\right\} + \E\left[\, y\left(t_{2}\right) \mid x_{1}^k, x^{-k} \,\right] \\
\leq \E\left[\, y\left(t_{1}\right) \mid x_{1}^k, x^{-k} \,\right] \leq  \\ 
\inf_{x_{1}^k \leq x^{k}_{2}}\left\{\overline{B}\left(t_{1}, x^{k}_{2}, x^{-k}\right) - \underline{B}\left(t_{2}, x^{k}_{2}, x^{-k}\right)\right\} + \E\left[\, y\left(t_{2}\right) \mid x_{1}^k, x^{-k} \,\right]
\end{gathered}\label{spmiv2a}
\end{equation}
and
\begin{equation}
\begin{gathered}
\E\left[\, y\left(t_{1}\right) \mid x_{1}^k, x^{-k} \,\right] - \inf_{x_{1}^k \leq x^{k}_{2}}\left\{\overline{B}\left(t_{1}, x^{k}_{2}, x^{-k}\right) - \underline{B}\left(t_{2}, x^{k}_{2}, x^{-k}\right)\right\} \\
\leq \E\left[\, y\left(t_{2}\right) \mid x_{1}^k, x^{-k} \,\right] \leq  \\ 
\E\left[\, y\left(t_{1}\right) \mid x_{1}^k, x^{-k} \,\right] - \sup_{x^{k}_{2} \leq x_{1}^k}\left\{\underline{B}\left(t_{1},x^{k}_{2}, x^{-k}\right) - \overline{B}\left(t_{2}, x^{k}_{2}, x^{-k}\right)\right\}
\end{gathered}\label{spmiv2b}
\end{equation}
From \eqref{spmiv na 1}--\eqref{spmiv2b}, it is clear that
\begin{gather*}
\max\left\{\underline{B}\left(t_{1}, x_{1}^k, x^{-k}\right), \sup_{x^{k}_{2} \leq x_{1}^k}\left\{\underline{B}\left(t_{1},x^{k}_{2}, x^{-k}\right) - \overline{B}\left(t_{2}, x^{k}_{2}, x^{-k}\right)\right\} + \underline{B}\left(t_{2},x_{1}^k, x^{-k}\right)\right\} \\
\leq \E\left[\, y\left(t_{1}\right) \mid x_{1}^k,x^{-k} \,\right] \leq \\
\min\left\{\overline{B}\left(t_{1}, x_{1}^k, x^{-k}\right), \inf_{x_{1}^k \leq x^{k}_{2}}\left\{\overline{B}\left(t_{1}, x^{k}_{2}, x^{-k}\right) - \underline{B}\left(t_{2}, x^{k}_{2}, x^{-k}\right)\right\} + \overline{B}\left(t_{2},x_{1}^k, x^{-k}\right)\right\}
\end{gather*}
and
\begin{gather*}
\max\left\{\underline{B}\left(t_{2}, x_{1}^k, x^{-k}\right),  \underline{B}\left(t_{1},x_{1}^k, x^{-k}\right) - \inf_{x_{1}^k \leq x^{k}_{2}}\left\{\overline{B}\left(t_{1}, x^{k}_{2}, x^{-k}\right) - \underline{B}\left(t_{2}, x^{k}_{2}, x^{-k}\right)\right\}\right\} \\
\leq \E\left[\, y\left(t_{2}\right) \mid x_{1}^k,x^{-k} \,\right] \leq \\
\min\left\{\overline{B}\left(t_{2}, x_{1}^k, x^{-k}\right), \overline{B}\left(t_{1},x_{1}^k, x^{-k}\right) - \sup_{x^{k}_{2} \leq x_{1}^k}\left\{\underline{B}\left(t_{1},x^{k}_{2}, x^{-k}\right) - \overline{B}\left(t_{2}, x^{k}_{2}, x^{-k}\right)\right\} \right\}
\end{gather*}
must hold. I show that these bounds are feasible, i.e., consistent with \eqref{spmiv 1}, whence it follows that they are sharp. Consider the following events:
\begin{equation}
\begin{aligned}\label{L1}
& \max\left\{\underline{B}\left(t_{1}, x_{1}^k, x^{-k}\right), \sup_{x^{k}_{2} \leq x_{1}^k}\left\{\underline{B}\left(t_{1},x^{k}_{2}, x^{-k}\right) - \overline{B}\left(t_{2}, x^{k}_{2}, x^{-k}\right)\right\} + \underline{B}\left(t_{2},x_{1}^k, x^{-k}\right)\right\} \\ 
& = \sup_{x^{k}_{2} \leq x_{1}^k}\left\{\underline{B}\left(t_{1},x^{k}_{2}, x^{-k}\right) - \overline{B}\left(t_{2}, x^{k}_{2}, x^{-k}\right)\right\} + \underline{B}\left(t_{2},x_{1}^k, x^{-k}\right) > \underline{B}\left(t_{1}, x_{1}^k, x^{-k}\right)
\end{aligned}
\end{equation}
\begin{equation}
\begin{aligned}\label{U1}
& \min\left\{\overline{B}\left(t_{1}, x_{1}^k, x^{-k}\right), \inf_{x_{1}^k \leq x^{k}_{2}}\left\{\overline{B}\left(t_{1}, x^{k}_{2}, x^{-k}\right) - \underline{B}\left(t_{2}, x^{k}_{2}, x^{-k}\right)\right\} + \overline{B}\left(t_{2},x_{1}^k, x^{-k}\right)\right\} \\ 
& = \inf_{x_{1}^k \leq x^{k}_{2}}\left\{\overline{B}\left(t_{1}, x^{k}_{2}, x^{-k}\right) - \underline{B}\left(t_{2}, x^{k}_{2}, x^{-k}\right)\right\} + \overline{B}\left(t_{2},x_{1}^k, x^{-k}\right) < \overline{B}\left(t_{1}, x_{1}^k, x^{-k}\right)
\end{aligned}
\end{equation}
\begin{equation}
\begin{aligned}\label{L2}
& \max\left\{\underline{B}\left(t_{2}, x_{1}^k, x^{-k}\right),  \underline{B}\left(t_{1},x_{1}^k, x^{-k}\right) - \inf_{x_{1}^k \leq x^{k}_{2}}\left\{\overline{B}\left(t_{1}, x^{k}_{2}, x^{-k}\right) - \underline{B}\left(t_{2}, x^{k}_{2}, x^{-k}\right)\right\}\right\} \\ 
& = \underline{B}\left(t_{1},x_{1}^k, x^{-k}\right) - \inf_{x_{1}^k \leq x^{k}_{2}}\left\{\overline{B}\left(t_{1}, x^{k}_{2}, x^{-k}\right) - \underline{B}\left(t_{2}, x^{k}_{2}, x^{-k}\right)\right\} > \underline{B}\left(t_{2}, x_{1}^k, x^{-k}\right)
\end{aligned}
\end{equation}
\begin{equation}
\begin{aligned}\label{U2}
& \min\left\{\overline{B}\left(t_{2}, x_{1}^k, x^{-k}\right), \overline{B}\left(t_{1},x\right) - \sup_{x^{k}_{2} \leq x_{1}^k}\left\{\underline{B}\left(t_{1},x^{k}_{2}, x^{-k}\right) - \overline{B}\left(t_{2}, x^{k}_{2}, x^{-k}\right)\right\} \right\} \\ 
& = \overline{B}\left(t_{1},x_{1}^k, x^{-k}\right) - \sup_{x^{k}_{2} \leq x_{1}^k}\left\{\underline{B}\left(t_{1},x^{k}_{2}, x^{-k}\right) - \overline{B}\left(t_{2}, x^{k}_{2}, x^{-k}\right)\right\} < \overline{B}\left(t_{2}, x_{1}^k, x^{-k}\right)
\end{aligned}
\end{equation}
It is easy to show that $\eqref{L1} \Longrightarrow \neg \eqref{L2}$; thus, the lower bounds in \eqref{spmiv po 1} and \eqref{spmiv po 2} are consistent with \eqref{spmiv 1}. Similarly, $\eqref{U1} \Longrightarrow \neg \eqref{U2}$, and so the upper bounds in \eqref{spmiv po 1} and \eqref{spmiv po 2} are consistent with \eqref{spmiv 1}.
\end{proof}

\begin{proof}[Proof of proposition \ref{SPMIV 3}]

The proof of this result is analogous to that of proposition 9.2 of \cite{manski2003partial}, where he derived sharp bounds on average potential outcomes using monotone treatment response and a monotone instrumental variable. Addressing the first portion of the claim, proposition \ref{spm bounds 2} shows that 
\begin{equation}
\begin{gathered}
\underline{\Delta B}\left(t_{1},t_{2},x_{1}^k,x^{-k}\right)
\leq \E\left[\, y\left(t_{1}\right) - y\left(t_{2}\right) \mid x_{1}^k,x^{-k} \,\right] \leq 
\overline{\Delta B}\left(t_{1},t_{2},x^{k}_{1},x^{-k}\right)
\end{gathered}
\end{equation}
and \ref{SPMIV} implies that 
\begin{equation}
\begin{gathered}
\E\left[\, y\left(t_{1}\right) - y\left(t_{2}\right) \mid x_{2}^k,x^{-k} \,\right] \leq \E\left[\, y\left(t_{1}\right) - y\left(t_{2}\right) \mid x_{1}^k,x^{-k} \,\right] \leq \E\left[\, y\left(t_{1}\right) - y\left(t_{2}\right) \mid x_{3}^k,x^{-k} \,\right]
\end{gathered}
\end{equation}
when $x_{2}\leq x_{1} \leq x_{3}$. The result follows. The derivation is analogous for the second portion of the claim.

\end{proof}

\begin{proof}[Proof of proposition \ref{te si}]

For a lattice $T = \left\{t_{1}, t_{2}, t_{1}\vee t_{2}, t_{1}\wedge t_{2}\right\}$ which is not a chain, \ref{SPM} implies the following inequalities:
\begin{gather*}
y\left(t_{2}\right) - y\left(t_{1} \wedge t_{2}\right) \leq y\left(t_{1} \vee t_{2}\right) - y\left(t_{1}\right) \label{SPM ineq a} \\
y\left(t_{1}\right) - y\left(t_{1} \wedge t_{2}\right) \leq y\left(t_{1} \vee t_{2}\right) - y\left(t_{2}\right) \label{SPM ineq b} \\
\end{gather*}
Since these inequalities hold for all individuals, they imply the following first--order stochastic dominance relationships:
\begin{align*}
F_{t_{1}\vee t_{2},t_{1}}\left(w\right) & \leq F_{t_{2},t_{1}\wedge t_{2}}\left(w\right) \\
F_{t_{1}\vee t_{2},t_{2}}\left(w\right) & \leq F_{t_{1},t_{1}\wedge t_{2}}\left(w\right)
\end{align*}
These imply the following bounds on quantile functions:
\begin{align}
F_{t_{2}, t_{1}\wedge t_{2}}^{-1}\left(q\right) & \leq F_{t_{1}\vee t_{2}, t_{1}}^{-1}\left(q\right) \label{sd ineq 1} \\
F_{t_{1}, t_{1}\wedge t_{2}}^{-1}\left(q\right) & \leq F_{t_{1}\vee t_{2}, t_{2}}^{-1}\left(q\right) \label{sd ineq 2}
\end{align}
Combining \eqref{sd ineq 1} and \eqref{sd ineq 2} with the bounds from \cite{williamson1990probabilistic} yields the result.

\end{proof}

\begin{proof}[Proof of proposition \ref{qspmiv}]
Trivial.
\end{proof}

\end{section}

\end{document}